\documentclass[aps,amsmath,amssymb,twocolumn,nofootinbib,prx,superscriptaddress]{revtex4-2}

\usepackage{graphicx,tikz}% Include figure files
\usepackage{dcolumn}% Align table columns on decimal point
\usepackage{bm}% bold math
\usepackage{epsfig,hyperref,amsthm}
\usepackage{mathtools}
\usepackage{color,fancybox}
\usepackage{colordvi}
\usepackage{relsize}
\usepackage{accents}
\usepackage{wasysym} 
\usepackage{enumitem}
\usepackage{mathtools,slashed}
\usepackage{times}
\usepackage{newtxmath}

\hypersetup{
	colorlinks=true,
	linkcolor=CitingColor,
	citecolor=CitingColor,
	urlcolor=CitingColor
}

\DeclareMathAlphabet\mathbfcal{OMS}{cmsy}{b}{n}
\newcommand{\ket}[1]{\ensuremath{|#1\rangle}}
\newcommand{\bra}[1]{\ensuremath{\langle #1|}}

\newcommand{\proj}[1]{\ket{#1}\bra{#1}}

\newcommand{\id}{\mathbb{I}}

\newcommand{\MA}{{\bf E}}

\newtheorem{alemma}{Lemma}[section]

\newtheorem{atheorem}[alemma]{Theorem}

\newtheorem{result}{Theorem}

\newtheorem{observation}[result]{Observation}
\newtheorem{question}{Question}

\newtheorem{definition}{Definition}

\definecolor{nred}{rgb}{0.9,0.1,0.1}
\definecolor{nblack}{rgb}{0,0,0}
\definecolor{nblue}{rgb}{0.2,0.2,0.8}
\definecolor{ngreen}{rgb}{0.2,0.6,0.2}
\definecolor{ublue}{rgb}{0,0,0.5}

\definecolor{pur}{rgb}{0.75,0,0.75}
\definecolor{nngrn}{rgb}{0,0.5,0.5}
\definecolor{CitingColor}{rgb}{0,0.3,1}

\newcommand{\blu}{\color{nblue}}

\begin{document}
\title{Characterisation and fundamental limitations of irreversible stochastic steering distillation}

\author{Chung-Yun Hsieh}
\affiliation{H. H. Wills Physics Laboratory, University of Bristol, Tyndall Avenue, Bristol, BS8 1TL, UK}
\author{Huan-Yu Ku}
\email{huan-yu.ku@oeaw.ac.at}
\affiliation{Department of Physics, National Taiwan Normal University, Taipei 11677, Taiwan}
\affiliation{Faculty of Physics, University of Vienna, Boltzmanngasse 5, 1090 Vienna, Austria}
\affiliation{Institute for Quantum Optics and Quantum Information (IQOQI), Austrian Academy of Sciences, Boltzmanngasse 3, 1090 Vienna, Austria}
\author{Costantino Budroni}
\affiliation{Department of Physics ``E. Fermi'' University of Pisa, Largo B. Pontecorvo 3, 56127 Pisa, Italy}

\date{\today}

\begin{abstract}
Steering resources, central for quantum advantages in one-sided device-independent quantum information tasks, can be enhanced via local filters.
Recently, reversible steering conversion under local filters has been fully characterised.
Here, we solve the problem in the {\em irreversible} scenario, which 
leads to a complete understanding of stochastic steering distillation.
This result also provides an operational interpretation of the max-relative entropy  as the optimal filter success probability. 
We further show that all steering measures can be used to quantify measurement incompatibility in certain stochastic steering distillation scenarios.
Finally, for a broad class of steering robustness measures, we show that their maximally achievable values in stochastic steering distillation are always {\em upper bounded} by different types of incompatibility robustness measures.
Hence, measurement incompatibility sets the fundamental limitations for stochastic steering distillation.
\end{abstract}

\maketitle

\section{introduction}
For each quantum information task, there is a specific quantum resource necessary for its successful 
implementation \cite{ChitambarRMP2019}.
For instance, perfect teleportation~\cite{Bennett93} and super-dense coding~\cite{Bennett92} require maximally entangled states. 
Also, perfect quantum memories and quantum communication protocols need noiseless identity channels~\cite{wilde_2017}. 
However, physical systems are usually noisy and not perfectly under control. 
This leads to the question of how these resources can be manipulated and enhanced, i.e., the problem of resource {\em distillation}.

Since the inception of entanglement distillation~\cite{Bennett1996,Horodecki1999}, the idea of distillation has 
been applied to various physical phenomena, such as quantum communication~\cite{Regula2021,Regula2021PRL}, 
thermodynamics~\cite{Lostaglio2019,Purity-review}, informativeness of measurements~\cite{Paul2022}, 
nonlocality~\cite{Gisin1996PLA,Forster2009PRL,Naik2023PRL}, and quantum steering~\cite{Nery2020,Liu2022,Ku2022NC,Ku2023}. 
In particular, quantum steering is inextricably linked to one-sided device-independent quantum information
processing~\cite{ShinLiang2016PRL,Skrzypczyk2018,Chen2021robustselftestingof,Branciard2012, Piani2015,Zhao2020,Ku2022PRXQ}.
However, if compared with other quantum resources, steering distillation has been studied only very 
recently~\cite{Nery2020,Ku2022NC,Ku2023,Liu2022}.
Moreover, {\em stochastic} steering distillation is using local filters to enhance quantum steering, which is not only experimentally feasible~\cite{Nery2020}, but also has strong implications for incompatible measurements~\cite{Ku2022NC,Ku2023}.
Still, the question of characterising \textit{irreversible} steering conversion remains open.
Up to now, the most timely characterisation~\cite{Nery2020,Ku2022NC,Ku2023} addresses the {\em reversible} case, 
in which one can undo the local filter with non-vanishing success probability.
It is then essential to characterise local filters that {\em cannot} be undone in order to completely understand 
stochastic steering distillation.

In this work, we prove the first necessary and sufficient characterisation of irreversible steering distillation under local filters.
Furthermore, we provide a systematic way to quantify measurement incompatibility~\cite{Otfried2021Rev} by the optimally distillable steering. 
Finally, by uncovering the relation between different types of robustness measures of steering and incompatibility, we show that incompatibility sets the fundamental limitations for stochastic steering distillation.

\section{Preliminary Notions}
We start with specifying notations.
For a {\em non-full-rank} positive operator $O\ge0$, one cannot define its inverse in general.
While one can effectively do so when we restrict to its {\em support} ${\rm supp}(O)$, 
which is the subspace spanned by its eigenstates with strictly positive eigenvalues (it is also called {\em range} and denoted by ${\rm ran}(O)$). 
Then, $O$ is effectively full-rank in the subspace ${\rm supp}(O)$, and its inverse can be well-defined as given by $O|_{{\rm supp}(O)}^{-1}\oplus 0_{\perp {\rm supp}(O)}$.
A conventional abbreviation in the literature is
$
O^{-1}\coloneqq O|_{{\rm supp}(O)}^{-1}\oplus 0_{\perp {\rm supp}(O)}.
$

\subsection{Quantum Steering}

As an intermediate quantum resource between nonlocality and entanglement, quantum steering refers to one agent ($A$) remotely preparing states for a spatially separated agent ($B$) by $A$'s local measurement acting on their shared bipartite state $\rho_{AB}$ with classical communication to $B$~\cite{Wiseman2007PRL,Cavalcanti2016,UolaRMP2020,XiangPRXQ2022} (see Fig.~\ref{Fig:steering} for a schematic illustration).
Formally, $A$'s measurements are described by a set of {\em positive operator-valued measures} (POVMs)~\cite{QIC-book} ${\bf E}\coloneqq\{E_{a|x}\}_{a,x}$ such that $E_{a|x}\ge0$ for every $a,x$ and $\sum_{a}E_{a|x} = \id_A$ for every $x$, where $\id_A$ is the identity operator in the system $A$.
We call ${\bf E}$ a {\em measurement assemblage}.
After $A$ locally measures $\rho_{AB}$, $B$ obtains the post-measurement (un-normalised) state 
\begin{align}
\sigma_{a|x}\coloneqq{\rm tr}_A[(E_{a|x}\otimes\id_B)\rho_{AB}].
\end{align}
The collection of un-normalised states ${\bm\sigma}\coloneqq\{\sigma_{a|x}\}_{a,x}$ is called a {\em state assemblage}.
Note that $\rho_{\bm\sigma}\coloneqq\sum_a\sigma_{a|x}$ is a state independent of $x$, which is due to the no-signalling condition.

The {\em classical} cases correspond to some state assemblage ${\bm\tau}=\{\tau_{a|x}\}_{a,x}$ that can be simulated by a {\em local-hidden-state} (LHS) model:
\begin{align}
\tau_{a|x} = \sum_\lambda P(\lambda)P(a|x,\lambda)\rho_\lambda,
\end{align}
where $\{P(\lambda)\}_\lambda,\{P(a|x,\lambda)\}_{a,x,\lambda}$ are (conditional) probability distributions, and $\{\rho_\lambda\}_\lambda$ are pre-existing states.
When a state assemblage ${\bm\sigma}$ does not admit a LHS model, denoted as \mbox{${\bm\sigma}\notin{\bf LHS}$,}  we say that ${\bm\sigma}$ is {\em steerable}.

\subsection{Measurement Incompatibility}
Measurement incompatibility is a key ingredient of many quantum phenomena such as uncertainty relations ~\cite{Busch2014RMP,Otfried2021Rev}, Bell nonlocality~\cite{Wolf2009PRL,QuintinoPRL2019,Zhao2023Optica},  contextuality~\cite{BudroniRMP2022}, as well as steering \cite{UolaRMP2020}. 
To formally define measurement incompatibility, we say a measurement assemblage ${\bf M}=\{M_{a|x}\}_{a,x}$ is {\em jointly measurable}, or {\em compatible}, if there exists a single POVM $\{G_\lambda\}_\lambda$ and conditional probability distributions $\{P(a|x,\lambda)\}_{a,x,\lambda}$ such that
\begin{align}
M_{a|x} = \sum_\lambda P(a|x,\lambda)G_\lambda.
\end{align}
Physically, they can be simulated by a single measurement $\{G_\lambda\}_\lambda$ with classical post-processing. 
For a measurement assemblage ${\bf E}$, we write ${\bf E}\in{\bf JM}$ whenever it is jointly measurable. 
Otherwise, it is {\em incompatible} and denoted by ${\bf E}\notin{\bf JM}$.

\begin{figure}
\scalebox{1.0}{\includegraphics{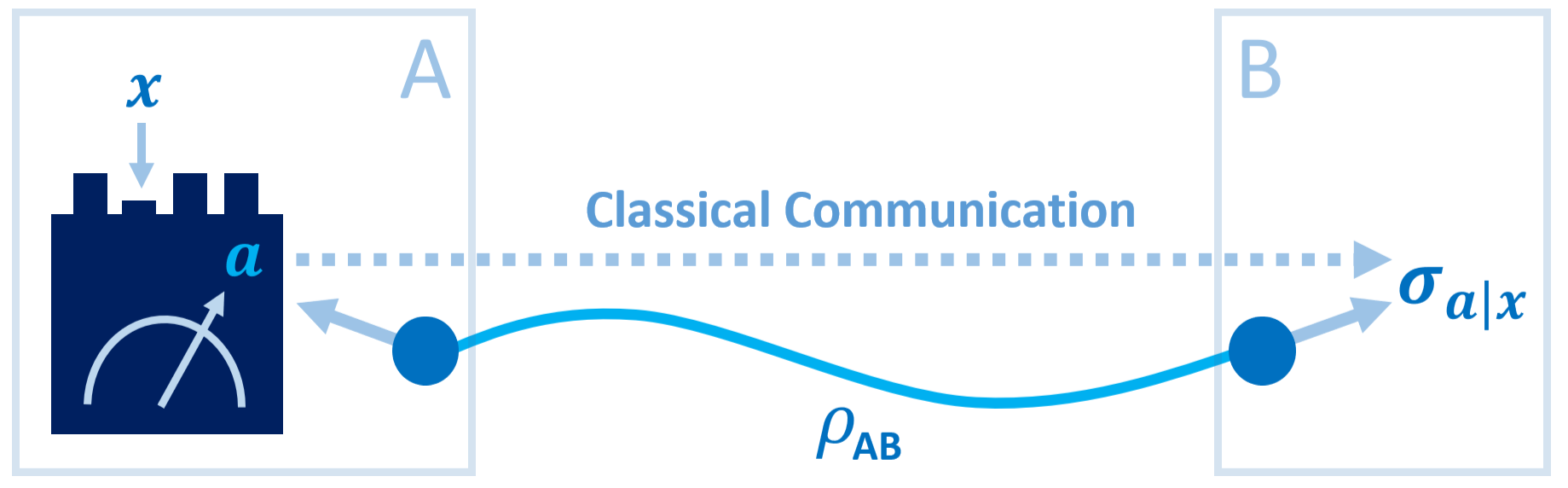}}
\caption{
{\bf A schematic interpretation of quantum steering.}
Quantum steering refers to one party $A$ remotely preparing states for a spatially separated party $B$ by some local measurement in $A$ (acting on the shared state $\rho_{AB}$) and classical communication from $A$ to $B$. 
}
\label{Fig:steering}
\end{figure}

\section{Main Results}

\subsection{Steering Distillation Under Irreversible Local Filters}
We first recall the notion of {\em steering-equivalent observable} (SEO)~\cite{Uola2015PRL,Kiukas2017PRA}, which mathematically connects quantum steering and measurement incompatibility.
Given a state assemblage ${\bm\sigma}$, its SEO is uniquely defined by
\begin{align}
{\bf B}^{({\bm\sigma})}\coloneqq\left\{\sqrt{\rho_{\bm\sigma}}^{\;-1}\sigma_{a|x}\sqrt{\rho_{\bm\sigma}}^{\;-1}\right\}_{a,x},
\end{align}
which are POVMs in the space ${\rm supp}(\rho_{\bm\sigma})$.
As shown in Refs.~\cite{Uola2015PRL,Kiukas2017PRA},
{\em ${\bm\sigma}\notin{\bf LHS}$ if and only if ${\bf B}^{({\bm\sigma})}\notin{\bf JM}$.}

As previously mentioned, only the agent $B$ is trusted in a steering experiment.
According to the one-sided device-independent perspective, only the system of $B$ is characterised. 
Hence, it is essential to understand how the local agent $B$ can manipulate state assemblages with their trusted devices.
It turns out that the SEO provides a clean way to {\em classify} steering assemblages under filter operations locally implemented by $B$.
More precisely, two state assemblages have the same SEO (up to a unitary degree of freedom) if and only if they can be transformed to each other by local filter operations~\cite{Ku2022NC}.
In other words, there exist local filters in {\em both} directions, and, probabilistically, one is able to {\em reverse} the effect of local filters.
The question, thus, remains open of how to characterise irreversible local filter transformations.
That is, can we compare state assemblages by the actions of local filters that are {\em irreversible} even {\em probabilistically}? 
To this end, we introduce the {\em SEO ordering}:
\begin{definition}
For state assemblages ${\bm\sigma}$ and ${\bm\tau}$, we write
\begin{align}
{\bm\sigma}\succ_{\rm SEO}{\bm\tau}
\end{align}
if and only if there is a unitary $U$ with \mbox{${\rm supp}(\rho_{\bm\tau})\subseteq{\rm supp}(U\rho_{\bm\sigma}U^\dagger)$} such that
\begin{align}\label{eq:def_>}
{\bm\tau} = \sqrt{\rho_{\bm\tau}} U {\bf B}^{({\bm\sigma})}U^\dagger\sqrt{\rho_{\bm\tau}}.
\end{align}
\end{definition}
In other words, ${\bm\sigma}\succ_{\rm SEO}{\bm\tau}$ whenever ${\bm\tau}$ can be mathematically induced by the SEO of ${\bm\sigma}$. 
One can check that $\succ_{\rm SEO}$ is a {\em preorder}, i.e., a reflexive and transitive homogeneous relation.
Also, the SEO ordering generalises the notion of {\em SEO equivalence class}, introduced in Ref.~\cite{Ku2022NC}, to the one-way, asymmetric regime.
More precisely, ${\bm\sigma},{\bm\tau}$ are {\em SEO equivalent}, denoted by ${\bm\sigma}\sim_{\rm SEO}{\bm\tau}$, if and only if there exists some unitary $U$ such that $U{\bf B}^{({\bm\sigma})}U^\dagger={\bf B}^{({\bm\tau})}.$
It is then straightforward to see that
${\bm\sigma}\sim_{\rm SEO}{\bm\tau}$ {\em if and only if }
${\bm\sigma}\succ_{\rm SEO}{\bm\tau}$ and ${\bm\tau}\succ_{\rm SEO}{\bm\sigma}.$

An immediate question now is whether this definition is operationally relevant.
As we shall demonstrate later, this mathematical notion is {\em operationally equivalent to} the following type of {\em local filter} operations:
\begin{align}
\sigma_{a|x}\mapsto \frac{K\sigma_{a|x}K^\dagger}{p_{\rm succ}},
\end{align}
where the operator $K$ with the condition $K^\dagger K\le\id$ describes the local filter operation $(\cdot)\mapsto K(\cdot)K^\dagger$ by a single Kraus operator, and $p_{\rm succ}\coloneqq{\rm tr}\left(K\rho_{\bm\sigma}K^\dagger\right)$ is the probability for a successful filter with the input ${\bm\sigma}$. 
We denote the set of all such local filters, i.e., those with a single Kraus operator, as ${\rm LF_1}$~\cite{Ku2022NC}. 
Conditioned on the successful filtering outcomes, \mbox{$\tau_{a|x} = K\sigma_{a|x}K^\dagger/p_{\rm succ}$} is the state assemblage that one obtains.
When there exists at least one such local filter achieving this transformation, with {\em non-vanishing} success probability, we write ${\bm\sigma}\xrightarrow{\rm LF_1}{\bm\tau}$.
Also, let $p_{\rm succ}^{\rm max}({\bm\sigma}\xrightarrow{\rm LF_1}{\bm\tau})$ be the highest success probability among all such ${\rm LF_1}$ transformations.

It is helpful, at this point, to recall the definition of {\em max-relative entropy}~\cite{Datta2009}.
For two states $\rho,\eta$, define
\begin{align}
D_{\rm max}(\eta\,\|\,\rho)\coloneqq\log_2\min\{\lambda\ge0\,|\,\eta\le\lambda \rho \}
\end{align}
if \mbox{${\rm supp}(\eta)\subseteq{\rm supp}(\rho)$;} otherwise, define $D_{\rm max}(\eta\,\|\,\rho)\coloneqq\infty$.
Then we have the following result:

\begin{result}\label{Result:LF1=SEO-ordering}
Let ${\bm\sigma},{\bm\tau}$ be state assemblages.
Then 
\begin{align}
{\bm\sigma}\succ_{\rm SEO}{\bm\tau}\quad\text{if and only if}\quad{\bm\sigma}\xrightarrow{\rm LF_1}{\bm\tau}.
\end{align}
When such a local filter transformation exists, we have
\begin{align}\label{Eq:Psucc_Dmax}
p_{\rm succ}^{\rm max}({\bm\sigma}\xrightarrow{\rm LF_1}{\bm\tau})=\sup_{U\in\mathcal{U}({\bm\sigma}\succ_{\rm SEO}{\bm\tau})}2^{-D_{\rm max}\left(\rho_{\bm\tau}\,\|\,U\rho_{\bm\sigma}U^\dagger\right)},
\end{align}
where $\mathcal{U}({\bm\sigma}\succ_{\rm SEO}{\bm\tau})$ is the set of all unitary operations consistent with a valid decomposition in Eq.~\eqref{eq:def_>}.
\end{result}
We detail the proof in Appendix~\ref{App:Proof-Result:LF1=SEO-ordering}. 
Also, see Fig.~\ref{Fig:Result1} for a schematic illustration.

As shown in Ref.~\cite{Ku2023}, optimal stochastic steering distillation can be achieved by using ${\rm LF_1}$ --- since no general local filter can distil more steerability than the optimal ${\rm LF_1}$ filter.
From this perspective, it is general enough to focus on ${\rm LF_1}$ for stochastic steering distillation.
Consequently, one can view Theorem~\ref{Result:LF1=SEO-ordering} as a complete characterisation of stochastic steering distillation --- since it provides a complete and analytical characterisation of irreversible steering conversion under ${\rm LF_1}$.
See also Observation~\ref{Obs} for further discussions.

As a direct corollary, Theorem~\ref{Result:LF1=SEO-ordering} reproduces the main result in Ref.~\cite{Ku2022NC} by considering reversible ${\rm LF_1}$; namely, one can check that: {\em ${\bm\sigma}\sim_{\rm SEO}{\bm\tau}$ $\Leftrightarrow$ ${\bm\sigma}\xrightarrow{\rm LF_1}{\bm\tau}$ $\&$ ${\bm\tau}\xrightarrow{\rm LF_1}{\bm\sigma}$ $\Leftrightarrow$ ${\bm\sigma}\succ_{\rm SEO}{\bm\tau}$ and ${\rm supp}(\rho_{\bm\sigma}),{\rm supp}(\rho_{\bm\tau})$ have the same dimension.}
In Ref.~\cite{Ku2022NC}, state assemblages are classified based on mutual convertibility under ${\rm LF_1}$; here, Theorem~\ref{Result:LF1=SEO-ordering} upgrades the steering classification into a preorder among state assemblages based on the most general conversion under ${\rm LF_1}$.

Rather unexpectedly, Theorem~\ref{Result:LF1=SEO-ordering} provides a new operational way to interpret max-relative entropy in a steering experiment.
Equation~\eqref{Eq:Psucc_Dmax} implies that max-relative entropy actually carries the information of the success probability for local filter transformations ${\bm\sigma}\xrightarrow{\rm LF_1}{\bm\tau}$.
This also explains the importance of the support condition from the success probability perspective.
In fact, the support condition implies a no-go result: {\em It is impossible to have ${\bm\sigma}\xrightarrow{\rm LF_1}{\bm\tau}$ if the rank of $\rho_{\bm\tau}$ is strictly higher than the one of $\rho_{\bm\sigma}$} --- since $D_{\rm max}\left(\rho_{\bm\tau}\,\middle\|\,U\rho_{\bm\sigma}U^\dagger\right)=\infty$ for all unitary $U$ in this case.

As a final remark, note that the formula of optimal success probability, i.e., Eq.~\eqref{Eq:Psucc_Dmax}, has a {\em prerequisite}: namely, we first need to know that such an ${\rm LF_1}$ filter {\em exists}.
For instance, local filters cannot create steerability when state assemblage is unsteerable. 
Hence, we do have $\rho_{\bm\sigma}$ and $\rho_{\bm\tau}$ with finite-valued max-relative entropy that no ${\rm LF_1}$ filter can convert $\bm{\sigma}\in {\bf LHS}$ to $\bm{\tau} \notin {\bf LHS}$.
Consequently, it is {\em insufficient} to merely check the value of max-relative entropy if one wants to know whether a local filter transformation exists or not.

\begin{figure}
\scalebox{1.0}{\includegraphics{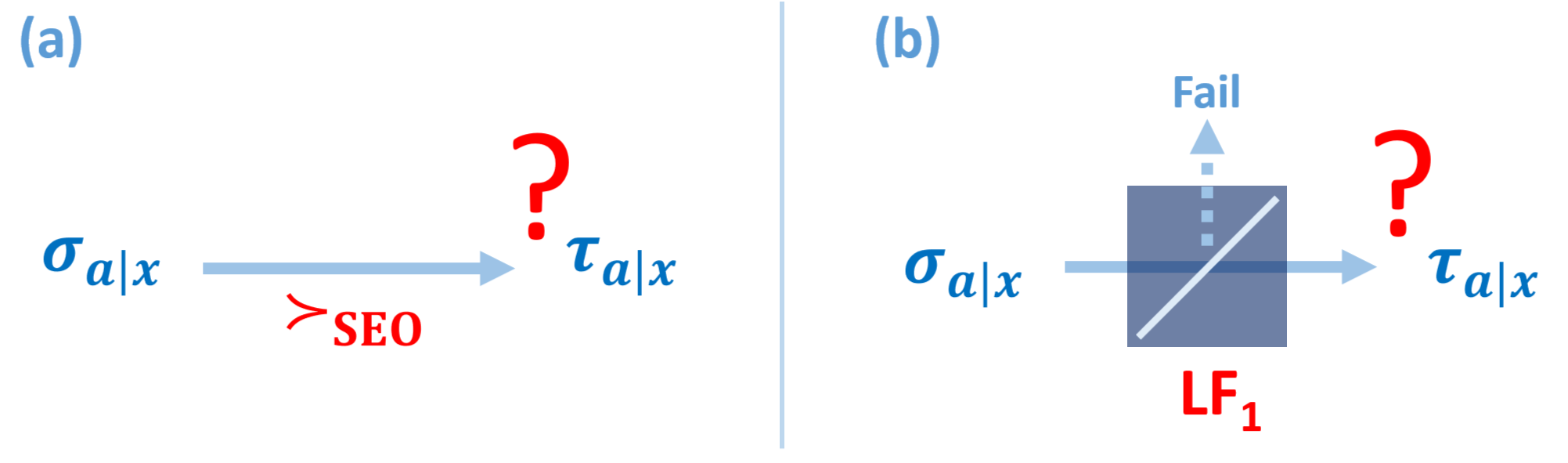}}
\caption{
{\bf Schematic interpretation of Theorem~\ref{Result:LF1=SEO-ordering}.}
Theorem~\ref{Result:LF1=SEO-ordering} considers two fundamental quantum information scenarios in a steering setup, which are ordering and one-way convertibility.
(a) To characterise and compare different steering resources beyond classification, we need an ordering structure among all state assemblages. This can be captured by the SEO ordering introduced in this work.
(b) An one-way convertibility problem aims to determine the possibility of transforming one object into another.
In this work, we focus on transformations under local filters in ${\rm LF_1}$, which are key ingredients in stochastic steering distillation~\cite{Ku2022NC,Ku2023}.
}
\label{Fig:Result1}
\end{figure}

\subsection{Examples}
As an illustrative example, consider a qubit-qutrit setting with the initial state $\rho_{AB}^{(v)} = v\proj{\phi_+^{\rm qubit}} + \mbox{(1-v)}\frac{\id_A}{2}\otimes\proj{2}$, where $\ket{\phi_+^{\rm qubit}}\coloneqq(\ket{00}+\ket{11})/\sqrt{2}$ is maximally entangled in the 'effective' qubit spanned by $\ket{0},\ket{1}$ in the qutrit $B$.
Suppose the initial state assemblage is
\mbox{$
\sigma^{(v)}_{a|x} = {\rm tr}_A\left[\left(E_{a|x}^{\rm Pauli}\otimes\id_B\right)\rho_{AB}^{(v)}\right],
$}
with
$
E_{0|0}^{\rm Pauli} = \proj{0} = \id_A-E_{1|0}^{\rm Pauli},
E_{0|1}^{\rm Pauli} = \proj{+} = \id_A-E_{1|1}^{\rm Pauli},
$
and $\ket{+}\coloneqq(\ket{0}+\ket{1})/\sqrt{2}$.
Namely, $E_{a|x}^{\rm Pauli}$'s are projective measurements corresponding to Pauli $X$ and $Z$ in the qubit system $A$.
By construction, the steerability of ${\bm\sigma}^{(v)}$ can be arbitrarily weak when $v\to0$.
Now, consider the ${\rm LF_1}$ filter $K = \proj{0}+\proj{1}$.
Then direct computation shows that, after this local filter and post-selection, we obtain
$
\sigma^{{\rm final}}_{a|x} = E_{a|x}^{\rm Pauli}/2,
$
which is maximally steerable in a two-qubit setting.
Note that this filter operation is not reversible, since ${\rm LF_1}$ cannot extend the final state assemblage to the degree of freedom spanned by $\ket{2}$ in $B$.
 This demonstrates how to use a genuinely irreversible ${\rm LF_1}$ filter to distil steering.

\section{Applications}

\subsection{Quantifying Measurement Incompatibility by Steering Distillation}
As an application, we show that steering distillation under local filters in ${\rm LF_1}$ can be used to {\em quantify} measurement incompatibility.
To this end, we introduce the {\em steering-induced incompatibility measure}:
\begin{equation}\label{Eq:steering-induced income measure}
I_S\left(\MA\right)\coloneqq\sup\left\{S\left({\bm\tau}\right)\,\middle|\,{\bm\sigma}\xrightarrow{\rm LF_1}{\bm\tau}, {\bf B}^{({\bm\sigma})} = \MA\right\},
\end{equation}
where the maximisation is taken over every state assemblage ${\bm\sigma}$ whose SEO is identical to $\MA$, and all other state assemblages that can be reached by them via ${\rm LF_1}$. 
When $S$ is a faithful steering measure (i.e., $S({\bm\sigma})=0$ if and only if ${\bm\sigma}\in{\bf LHS}$), it is not difficult to check that \mbox{$I_S\left(\MA\right)=0$} if and only if $\MA\in{\bf JM}$.
Adopting the resource theory settings considered in Refs.~\cite{Gallego2015PRX,Paul2019PRL}, we prove the following result:
\begin{result}\label{Results:IS is income monotone}
$I_{S}$ is an incompatibility monotone if $S$ is a steering monotone.
\end{result}
Detailed proof as well as a complete statement of this result can be found in Appendix~\ref{App:Result2}.
Theorem~\ref{Results:IS is income monotone} {\em quantitatively} bridges quantum steering and incompatibility in the most general setting via their corresponding resource theories~\cite{Gallego2015PRX,Paul2019PRL}.
Notably, if there is a steering measure with certain physical/operational meaning, one can obtain a measure for measurement incompatibility with a similar interpretation by using Theorem~\ref{Results:IS is income monotone}. 
For instance, if we consider the metrological steering monotone~\cite{Yadin2021NC}, one can then estimate measurement incompatibility via a metrologic task.
Theorem~\ref{Results:IS is income monotone} provides a full generality for this types of mapping.

\subsection{Fundamental Limitations on Steering Distillation}

It is useful to consider some specific measures to illustrate Theorem~\ref{Results:IS is income monotone}.
To this end, we focus on the {\em robustness-type measures}, which are widely used in quantum information theory~\cite{Designolle2019,Takagi2019,Hsieh2016PRA,Ku2018PRA,Tendick2023Quantum,Uola2015PRL,Haapasalo2015,Buscemi2020PRL,
Piani2015,Cavalcanti2016PRA,Chen2018PRA,Hsieh2022PRR,Hsieh2023-2,SDP-textbook}.
To illustrate what they are, now, we write $p{\bf E}+(1-p){\bf M}\coloneqq\{pE_{a|x}+(1-p)M_{a|x}\}_{a,x}$, and similar notations are also used for state assemblages.
The {\em generalised incompatibility robustness} is defined by~\cite{Uola2015PRL,Haapasalo2015,Buscemi2020PRL}
\begin{align}
{\rm IR}({\bf E})\coloneqq\inf\left\{t\ge0\,\middle|\,\frac{{\bf E}+t{\bf M}}{1+t}\in{\bf JM}\right\},
\end{align}
which is minimising over all possible measurement assemblages ${\bf M}$.
Physically, ${\rm IR}({\bf E})$ measures the smallest amount of {\em noise} needed to turn ${\bf E}$ compatible.
In other words, ${\bf M}$ takes the role of some noise being added to destroy ${\bf E}$'s incompatibility.
This type of measure can be defined for a wide range of physical phenomena with {\em different types} of noise.
For instance, the {\em generalised steering robustness}~\cite{Piani2015,Cavalcanti2016PRA} is 
\begin{align}
{\rm SR}({\bm\sigma})\coloneqq\inf\left\{t\ge0\,\middle|\,\frac{{\bm\sigma}+t{\bm\omega}}{1+t}\in{\bf LHS}\right\},
\end{align}
which minimises over every state assemblage ${\bm\omega}$.
As another example, the {\em consistent generalised steering robustness}~(see, e.g., Refs.~\cite{Cavalcanti2016PRA,Chen2018PRA}) is 
\begin{align}
{\rm SR}^{({\rm c})}({\bm\sigma})\coloneqq\inf\left\{t\ge0\,\middle|\,\frac{{\bm\sigma}+t{\bm\omega}}{1+t}\in{\bf LHS},\rho_{\bm\omega}=\rho_{\bm\sigma}\right\},
\end{align}
which can be viewed as a steering robustness with the `noise model' containing all ${\bm\omega}$ with the condition $\rho_{\bm\omega}=\rho_{\bm\sigma}$.
In general, one can consider different types of noise, and the associated steering robustness is called {\em consistent} if the condition $\rho_{\bm\omega}=\rho_{\bm\sigma}$ is always imposed.

It turns out that ${\rm LF_1}$ filters provide a general and quantitative way to link {\em different types} of robustness measures of incompatibility and steering.
To make this link precise, we introduce the following notion to describe different noise models.
Formally, we define 
a {\em noise model for quantum steering} as a map $\mathcal{N}_S$ that maps a state assemblage into {\em a set of} state assemblages; namely,
\mbox{${\bm\sigma}\mapsto\mathcal{N}_S({\bm\sigma}) = \{{\bm\omega}\}.$}
The set $\mathcal{N}_S({\bm\sigma}) = \{{\bm\omega}\}$ contains all possible ${\bm\omega}$ that can be used as noises to mix with the given ${\bm\sigma}$.
The {\em steering robustness subject to noise $\mathcal{N}_S$} is defined by
\begin{align}\label{Eq:SRNS}
{\rm SR}_{\mathcal{N}_S}({\bm\sigma})\coloneqq\inf\left\{t\ge0\,\middle|\frac{{\bm\sigma}+t{\bm\omega}}{1+t}\in{\bf LHS},{\bm\omega}\in\mathcal{N}_S({\bm\sigma})\right\}.
\end{align}
For instance, if $\mathcal{N}_S({\bm\sigma})$ is the set of all state assemblages, we revisit the generalised steering robustness ${\rm SR}$.
When $\mathcal{N}_S({\bm\sigma})=\{\text{${\bm\omega}$: state assemblage}\,|\,\rho_{\bm\omega}=\rho_{\bm\sigma}\}$, then we get the consistent generalised steering robustness ${\rm SR}^{({\rm c})}$.

Now, similarly, a {\em noise model for measurement incompatibility} is a map $\mathcal{N}_I$ that maps a measurement assemblage into {\em a set of} measurement assemblages; namely,
\mbox{$
{\bf E}\mapsto\mathcal{N}_I({\bf E}) = \{{\bf M}\}.
$}
The set $\mathcal{N}_I({\bf E}) = \{{\bf M}\}$ contains all possible ${\bf M}$ that can be used as noises to mix with ${\bf E}$.
The {\em incompatibility robustness subject to noise $\mathcal{N}_I$} is defined by 
\begin{align}\label{Eq:IRNI}
{\rm IR}_{\mathcal{N}_I}({\bf E})\coloneqq\inf\left\{t\ge0\,\middle|\,\frac{{\bf E}+t{\bf M}}{1+t}\in{\bf JM},\;{\bf M}\in\mathcal{N}_I({\bf E})\right\}.
\end{align}
If we set $\mathcal{N}_I({\bf E})=$ the set of all measurement assemblages, we revisit the generalised incompatible robustness ${\rm IR}$.

It turns out that ${\rm SR}_{\mathcal{N}_S}$ and ${\rm IR}_{\mathcal{N}_I}$ can be linked by SEO ordering. 
To this end, we say that $\mathcal{N}_I$ is {\em  SEO-included} by $\mathcal{N}_S$, and we denote it by $\mathcal{N}_I \subset_{\rm SEO} \mathcal{N}_S$, if it holds that 
\begin{align}\label{Eq:SEO-ordering-induced}
\sqrt{\eta}U\mathcal{N}_I\left({\bf M}\right)U^\dagger\sqrt{\eta}\subseteq\mathcal{N}_S\left(\sqrt{\eta}U{\bf M}U^\dagger\sqrt{\eta}\right)
\end{align}
for every measurement assemblage ${\bf M}$, state $\eta$, and unitary $U$ with the condition ${\rm supp}(\eta)\subseteq{\rm supp}(U\id_{\bf M}U^\dagger)$, where \mbox{$\id_{\bf M} \coloneqq \sum_aM_{a|x}\,\forall\,x$;} i.e., it is the identity of the (sub-)space that ${\bf M}$ lives in.
Note that the notation $O\mathcal{N}_I({\bf E})O^\dagger$ denotes the set
$\left\{O{\bf W}O^\dagger\,|\,{\bf W}\in\mathcal{N}_I({\bf E})\right\}$.
In other words, every {\em measurement} assemblage noise can induce some {\em state} assemblage noise via SEO ordering. 
Notice that $\subset_{\rm SEO}$ is {\em not} a preorder, since it always connects two different types of noises, so it is not homogeneous.
As one can check, the noise models for ${\rm IR}$ and ${\rm SR}$ satisfy this inclusion, and the same holds for the noise models for ${\rm IR}$  and ${\rm SR}^{({\rm c})}$.
Now we present the following result, which is proved in Appendix~\ref{App:Proof-Result:duality-noise-model}:
\begin{result}\label{Result:duality-noise-model}
Suppose $\mathcal{N}_I \subset_{\rm SEO} \mathcal{N}_S$.
Then for every state assemblage ${\bm\sigma}$, we have that
\begin{align}
\sup\left\{{\rm SR}_{\mathcal{N}_S}\left({\bm\tau}\right)\,\middle|\,{\bm\sigma}\xrightarrow{\rm LF_1}{\bm\tau}\right\} \le {\rm IR}_{\mathcal{N}_I}\left({\bf B}^{({\bm\sigma})}\right).
\end{align}
\end{result}
Interestingly, Theorem~\ref{Result:duality-noise-model} provides fundamental limitations on steering distillation under local filters in ${\rm LF_1}$ when we measure steerability via ${\rm SR}_{\mathcal{N}_S}$ --- namely, its value can {\em never} go beyond the incompatibility robustness ${\rm IR}_{\mathcal{N}_I}$.
In other words, as long as the SEO inclusion holds, the distillable steerability is always controlled and limited by the amount of incompatibility of the corresponding SEO.
This finding reveals a foundational hierarchy between steering and incompatibility in the most general setting.

\subsection{Quantifying Steering Distillation by Measurement Incompatibility}
A very natural question raised by Theorem~\ref{Result:duality-noise-model} is: {\em when can the upper bound be saturated?}
It turns out that it is closely related to the consistency condition for the steering robustness.
As an explicit example, we consider the relation of ${\rm SR},{\rm SR^{(c)}}$, and ${\rm IR}$ through steering distillation tasks.
Combining our finding and Refs.~\cite{Ku2022NC,Chen2018PRA}, we have the following observation (in what follows, ${\rm LF}$ denotes {\em general} local filters, i.e., completely-positive trace-non-increasing linear maps; we write ${\bm\sigma}\xrightarrow{\rm LF}{\bm\tau}$ if such a filter exists with non-vanishing success probability):
\begin{observation}\label{Obs}
For every ${\bm\sigma}$, we have that
\begin{align}
{\rm IR}\left({\bf B}^{({\bm\sigma})}\right)  &= \sup\left\{{\rm SR}({\bm\tau})\,\middle|\, {\bm\sigma}\xrightarrow{\rm LF}{\bm\tau}\right\}\nonumber\\
&= \sup\left\{{\rm SR}({\bm\tau})\,\middle|\, {\bm\sigma}\xrightarrow{\rm LF_1}{\bm\tau}\right\} = {\rm SR}^{({\rm c})}({\bm\sigma}).
\end{align}
\end{observation}
\begin{proof}
A direct computation shows that
\begin{align}
{\rm IR}\left({\bf B}^{({\bm\sigma})}\right)&\ge
\sup\left\{{\rm SR}({\bm\tau})\,\middle|\,{\bm\sigma}\xrightarrow{\rm LF_1}{\bm\tau}\right\}\nonumber\\
&\ge\sup\left\{{\rm SR}({\bm\tau})\,\middle|\,{\bm\sigma}\xrightarrow{\rm LF_1}{\bm\tau},{\bm\tau}\xrightarrow{\rm LF_1}{\bm\sigma}\right\}={\rm IR}\left({\bf B}^{({\bm\sigma})}\right).\nonumber
\end{align}
The first inequality is from Theorem~\ref{Result:duality-noise-model}, the second inequality is due to reducing the maximisation range, and the last equality is the main result of Ref.~\cite{Ku2022NC}.
Finally, since ${\rm LF_1}$ are optimal over all ${\rm LF}$ filters (see Ref.~\cite{Ku2023}) and  \mbox{${\rm IR}\left({\bf B}^{({\bm\sigma})}\right) = {\rm SR}^{({\rm c})}({\bm\sigma})$} (see Ref.~\cite{Chen2018PRA}), the result follows.
\end{proof}

This result has several implications. First, it provides an operational meaning to ${\rm SR^{(c)}}$ as the 
maximal steerability that can be extracted from ${\rm LF}$ filters, as it was previously noted for 
${\rm IR}$. In other words, ${\rm SR}^{({\rm c})}$ is a natural {\em steering distillation} measure. As a 
consequence, ${\rm SR^{(c)}}$ is a quantity that cannot be stochastically distilled by ${\rm LF}$ filters on the trusted side ($B$), as it provides the same value. 
Notice that, as a consequence, ${\rm SR^{(c)}}$ takes the same value on state assemblages with very different values of the generalised robustness ${\rm SR}$. 
In fact, in the same SEO equivalence class, there exist state assemblages with the maximal ${\rm SR}$ (corresponding to ${\rm IR}$) and also state assemblages with ${\rm SR}$ arbitrarily close to zero~\cite{Ku2022NC}, whereas ${\rm SR}^{({\rm c})}$ always provides the same value corresponding to the maximally distillable steerability, as quantified by ${\rm SR}$.
Finally, this finding shows that ${\rm LF}$ cannot outperform ${\rm LF_1}$ in stochastic steering distillation, as long as steerability is measured by ${\rm SR}$.
From this perspective, Theorem~\ref{Result:LF1=SEO-ordering} can be viewed as a complete characterisation of stochastic steering distillation in a general sense.

\section{Conclusion}
In this work, we prove the first complete characterisation of irreversible steering conversion under local filters in ${\rm LF_1}$. 
This finding also fully solves a major open question of Ref.~\cite{Ku2022NC}.
At the same time, we provide a general way of constructing incompatibility measures from steering measures and show how general incompatibility measures serve as upper bounds on the highest achievable steerability under local filters. 
As a consequence, fundamental limitations on irreversible stochastic steering distillation follow.
Finally, we provide an operational interpretation of the consistent steering robustness as a measure of steering distillability, rather than a direct measure of steerability.

Several questions remain open.
First, Ref.~\cite{Ku2023} reports that, rather surprisingly, measurement incompatibility {\em cannot} be distilled by local filters, even with non-physical filter operations.
Hence, despite their mathematical equivalence, quantum steering and measurement incompatibility behave in opposite ways with respect to stochastic distillation. 
It is then interesting to explore the difference between these two seemingly equivalent physical phenomena. Furthermore, it is interesting to study possible applications of stochastic steering distillation to thermodynamics and information transmission~\cite{BeyerPRL2019,Hsieh2020,Hsieh2021PRXQ,Hsieh2022,Hsieh2020,Stratton2023}.
Finally, from a practical perspective, it is useful to further study how to use stochastic steering distillation to improve steering inequality violation~\cite{CostaPRA2016,SkrzypczykPRL2018,JiPRL2022,LopeteguiPRXQ2022} and its activation~\cite{Hsieh2016,Quintino2016,KuPRA2023}, which can potentially enhance (one-sided) device-independent quantum information protocols.

\section*{Acknowledgements}
The authors acknowledge fruitful discussions with Paul Skrzypczyk and Roope Uola.
C.-Y. H.~is supported by the Royal Society through Enhanced Research Expenses (on grant NFQI) and the ERC Advanced Grant (FLQuant).
H.-Y.~K.~and C.~B.~are supported by the
Ministry of Science and Technology, Taiwan, (Grants No.~MOST 111-2917-I-564-005), the Austrian Science Fund (FWF) through Projects No.~ZK 3 (Zukunftskolleg), and No.~F7113 (BeyondC).

\appendix

\section{Proof of Theorem~\ref{Result:LF1=SEO-ordering}}\label{App:Proof-Result:LF1=SEO-ordering}
\begin{proof} ({\em direction} ``$\Rightarrow$'') Suppose ${\bm\sigma}\succ_{\rm SEO}{\bm\tau}$. 
Then we have 
\begin{align}\label{eq:sigma_tau}
\tau_{a|x} = \sqrt{\rho_{\bm\tau}} U \sqrt{\rho_{\bm\sigma}}^{\,-1}{\sigma}_{a|x}\sqrt{\rho_{\bm\sigma}}^{\,-1} U^\dagger \sqrt{\rho_{\bm\tau}}\quad\forall\,a,x
\end{align}
for some unitary $U$ achieving
${\rm supp}(\rho_{\bm\tau})\subseteq{\rm supp}(U\rho_{\bm\sigma}U^\dagger)$.
Equation \eqref{eq:sigma_tau} already suggests the correct Kraus operator, up to normalisation. 
Let us define 
\begin{align}\label{Eq:lambda_opt}
\lambda_{\rm opt}\coloneqq \min \left\{\lambda \geq 0\, |\,\rho_{\bm\tau}\leq \lambda  U\rho_{\bm\sigma}U^\dagger \right\}.
\end{align}
This implies 
\begin{align}
\sqrt{\rho_{\bm\sigma}}^{\,-1} U^\dagger \rho_{\bm\tau} U \sqrt{\rho_{\bm\sigma}}^{-1}\leq \lambda_{\rm opt} \id_{\bm\sigma},
\end{align}
where $\id_{\bm\sigma}$ is the identity of the subspace ${\rm supp}(\rho_{\bm\sigma})$. Notice that $\lambda_{\rm opt}$ exists and $1\le \lambda_{\rm opt} < \infty$,
because of the inclusion of the supports, and consequently
$D_{\rm max}(\rho_{\bm\tau}\,\|\,U\rho_{\bm\sigma}U^\dagger)=\log_2\lambda_{\rm opt}$.
Now, define the operator
\begin{align}\label{eq:L_def}
L\coloneqq \sqrt{\lambda_{\rm opt}}^{\,-1}\sqrt{\rho_{\bm\tau}} U \sqrt{\rho_{\bm\sigma}}^{\;-1}.
\end{align}
It is then straightforward to check that $L^\dagger L\le\id_{\bm\sigma}$, meaning that it is a local filter in ${\rm LF_1}$.
The success probability for the filter $L$ reads
\begin{align}\label{eq:p_succ_L}
p_{\rm succ} = {\rm tr}\left(L\rho_{\bm\sigma}L^\dagger\right)= \lambda_{\rm opt}^{-1} =  2^{-D_{\rm max}(\rho_{\bm\tau}\,\|\,U\rho_{\bm\sigma}U^\dagger)}>0,
\end{align}
which is non-vanishing since $\lambda_{\rm opt}<\infty$.
Hence, it provides the correct state assemblage, since
$
L\sigma_{a|x}L^\dagger/p_{\rm succ}
= \tau_{a|x}\;\forall\,a,x.
$
This means that ${\bm\sigma}\xrightarrow{\rm LF_1}{\bm\tau}$, as desired.

({\em  direction} ``$\Leftarrow$'') 
Suppose that ${\bm\sigma}\xrightarrow{\rm LF_1}{\bm\tau}$.
Then we have
\begin{align}\label{eq:tau_filter}
\tau_{a|x} = \frac{K\sqrt{\rho_{\bm\sigma}}{\bf B}_{a|x}^{({\bm\sigma})}\sqrt{\rho_{\bm\sigma}}K^\dagger}{p_{\rm succ}}
\end{align} 
with some local filter $K$ satisfying $K^\dagger K\le\id$ and \mbox{$p_{\rm succ} = {\rm tr}\left(K\rho_{\bm\sigma}K^\dagger\right)>0$.}
By the polar decomposition theorem~\cite{QIC-book}, there exist a unitary $U$ and positive operator $P\ge0$ such that 
\begin{align}
\frac{K\sqrt{\rho_{\bm\sigma}}}{\sqrt{p_{\rm succ}}}=UP,
\end{align}
where
\begin{align}
P = \sqrt{\left(\frac{K\sqrt{\rho_{\bm\sigma}}}{\sqrt{p_{\rm succ}}}\right)^\dagger\left(\frac{K\sqrt{\rho_{\bm\sigma}}}{\sqrt{p_{\rm succ}}}\right)}=\frac{\sqrt{\sqrt{\rho_{\bm\sigma}}K^\dagger K\sqrt{\rho_{\bm\sigma}}}}{\sqrt{p_{\rm succ}}}.
\end{align}
Using the relation $K^\dagger K\le\id$, we can write (with $p_{\rm succ}>0$)
\begin{align}\label{Eq:condition001}
P^2 \le\frac{\rho_{\bm\sigma}}{p_{\rm succ}},
\end{align} 
meaning that \mbox{${\rm supp}(P) = {\rm supp}(P^2)\subseteq{\rm supp}(\rho_{\bm\sigma})$.}
Together with Eq.~\eqref{eq:tau_filter}, we obtain
\begin{align}\label{Eq:condition002}
\rho_{\bm\tau} = UP\id_{\bm\sigma}PU^\dagger = UP^2U^\dagger.
\end{align}
Hence, $UPU^\dagger = \sqrt{\rho_{\bm\tau}}$, and we have \mbox{$\tau_{a|x} = \sqrt{\rho_{\bm\tau}}U{\bf B}_{a|x}^{({\bm\sigma})}U^\dagger\sqrt{\rho_{\bm\tau}}$} with ${\rm supp}(\rho_{\bm\tau})\subseteq{\rm supp}(U\rho_{\bm\sigma}U^\dagger)$, meaning that ${\bm\sigma}\succ_{\rm SEO}{\bm\tau}$.
This concludes the proof of the claim ${\bm\sigma}\succ_{\rm SEO}{\bm\tau} \Leftrightarrow {\bm\sigma}\xrightarrow{\rm LF_1}{\bm\tau}$.

({\em computation of maximal success probability}) 
Suppose ${\bm\sigma}\xrightarrow{\rm LF_1}{\bm\tau}$.
For {\em every} ${\rm LF_1}$ filter mapping ${\bm\sigma}\mapsto{\bm\tau}$ with a success probability $p_{\rm succ}$, Eqs.~\eqref{Eq:condition001} and~\eqref{Eq:condition002} imply that \mbox{$\rho_{\bm\tau}\le U\rho_{\bm\sigma}U^\dagger/p_{\rm succ}$} for some unitary $U$ consistent with the decomposition in Eq.~\eqref{eq:sigma_tau}.
Using Eq.~\eqref{Eq:lambda_opt}, we obtain
\mbox{$
p_{\rm succ}\le\lambda_{\rm opt}^{-1} =  2^{-D_{\rm max}(\rho_{\bm\tau}\,\|\,U\rho_{\bm\sigma}U^\dagger)}.
$}
Since this holds for {\em every possible} local filter in ${\rm LF_1}$ mapping ${\bm\sigma}\mapsto{\bm\tau}$, we conclude that
\begin{equation}
p_{\rm succ}^{\rm max}({\bm\sigma}\xrightarrow{\rm LF_1}{\bm\tau}) \le \sup_{U\in\mathcal{U}({\bm\sigma}\succ_{\rm SEO}{\bm\tau})}2^{-D_{\rm max}\left(\rho_{\bm\tau}\,\|\, U\rho_{\bm\sigma}U^\dagger\right)}.
\end{equation}
Finally, this upper bound can be attained, since every such $U$ can induce an ${\rm LF_1}$ filter with success probability $2^{-D_{\rm max}\left(\rho_{\bm\tau}\,\|\, U\rho_{\bm\sigma}U^\dagger\right)}$ by following the proof above Eq.~\eqref{eq:p_succ_L}.
The proof is thus completed.
\end{proof}

\section{Proof of Theorem~\ref{Results:IS is income monotone}}\label{App:Result2}
Before stating the formal result, we briefly state the allowed operations considered here.
In this work, we consider free operations of measurement incompatibility introduced in Ref.~\cite{Paul2019PRL}, which are mappings $\MA\mapsto\mathcal{L}_{\rm{inc}}(\MA)$ given by
\begin{align}
\mathcal{L}_{\rm{inc}}(\MA)_{a'|x'}\coloneqq\sum_{a,x,\omega}p(\omega)p(x|x',\omega)p(a'|a,x',\omega)E_{a|x}.
    \label{Eq: free incom Paul}
\end{align}
This operation can be seen as the classical post-processing with the pre-existing randomness $\omega$.
Let us also briefly recall the free operations in the resource theory of quantum steering~\cite{Gallego2015PRX}, which are mappings ${\bm\sigma}\mapsto\mathcal{E}_{\rm steer}({\bm\sigma})$ given by
\begin{align}
\mathcal{E}_{\rm steer}({\bm\sigma})_{a'|x'}\coloneqq\sum_{a,x,\omega}p(x|x',\omega)p(a'|a,x,x',\omega)\mathcal{E}_{\omega}(\sigma_{a|x}),
    \label{Eq: free steering}
\end{align}
where $\{\mathcal{E}_{\omega}\}_\omega$ form a {\em quantum instrument}, i.e., it is a set of completely-positive trace-non-increasing maps with the property that  $\sum_\omega \mathcal{E}_{\omega}$ is trace-preserving.

Now we provide a complete version of Theorem~\ref{Results:IS is income monotone} as the following theorem.
Recall that a steering measure $S$ is {\em faithful} if $S({\bm\sigma}) = 0$ if and only if ${\bm\sigma}\in{\bf LHS}$.

\begin{atheorem}
For an arbitrary faithful steering measure $S$, we have $I_{S}(\MA)=0$ if and only if $\MA\in{\bf JM}$.
Moreover, \mbox{$I_{S}\left[\mathcal{L}_{\rm{inc}}(\MA)\right]\le I_S(\MA)$} for every $\MA$ and free operation $\mathcal{L}_{\rm{inc}}$.
\end{atheorem}
\begin{proof}
It suffices to show the non-increasing property.
First, by using Theorem~\ref{Result:LF1=SEO-ordering}, we can rewrite Eq.~\eqref{Eq:steering-induced income measure} as
\begin{align}\label{Eq:Incomp_steer_measure alternative form}
&I_S(\MA) = \sup\left\{S({\bm\tau})\;\middle|\;{\bm\sigma}\succ_{\rm SEO}{\bm\tau},{\bf B}^{({\bm\sigma})} = \MA\right\}\nonumber\\
&=\sup\left\{S\left(\sqrt{\eta}U\MA U^\dagger\sqrt{\eta}\right)\,\middle|\,{\rm supp}(\eta)\subseteq{\rm supp}\left(U\id_\MA U^\dagger\right)\right\},
\end{align}
where the maximisation is taken over every state $\eta$ and unitary $U$ satisfying ${\rm supp}(\eta)\subseteq{\rm supp}\left(U\id_\MA U^\dagger\right)$.
Now, for a free operation of incompatibility as given in Eq.~\eqref{Eq: free incom Paul}, we have $\sqrt{\eta}U\mathcal{L}_{\rm{inc}}(\MA)U^\dagger\sqrt{\eta}= \mathcal{E}_{\rm steer}\left(\sqrt{\eta}U\MA U^\dagger\sqrt{\eta}\right)$ for some free operation of steering.
This is because
Eq.~\eqref{Eq: free incom Paul} can be viewed as special cases of Eq.~\eqref{Eq: free steering} with $\mathcal{E}_{\omega}=p(\omega)\mathcal{I}$, where $\mathcal{I}$ is the identity map.
Hence, we have
\begin{align}
&I_S\left[\mathcal{L}_{\rm{inc}}(\MA)\right]\nonumber\\
&=\sup\left\{S\left[\mathcal{E}_{\rm steer}\left(\sqrt{\eta}U\MA U^\dagger\sqrt{\eta}\right)\right]\;\middle|\;{\rm supp}(\eta)\subseteq{\rm supp}\left(U\id_\MA U^\dagger\right)\right\}\nonumber\\
&\le\sup\left\{S\left(\sqrt{\eta}U\MA U^\dagger\sqrt{\eta}\right)\;\middle|\;{\rm supp}(\eta)\subseteq{\rm supp}\left(U\id_\MA U^\dagger\right)\right\}\nonumber\\
&=I_S(\MA),
\end{align}
where we have used the non-increasing property of $S(\cdot)$ under $\mathcal{E}_{\rm steer}$ and Eq.~\eqref{Eq:Incomp_steer_measure alternative form}.
The proof is thus completed.
\end{proof}
When the given steering monotone $S$ is convex, the steering-induced incompatibility measure is also convex. 
More precisely, $I_{S}(\sum_i p_i\MA_i)\leq\sum_i p_i I_{S} (\MA_i)$ for every probability distribution $\{p_i\}_i$, as one can directly check by using Eq.~\eqref{Eq:Incomp_steer_measure alternative form}.
Also, we comment that, in the current setting, the free operations of incompatibility can be viewed as a subset of free operations of steering.
Finally, note that Eq.~\eqref{Eq:Incomp_steer_measure alternative form} implies that one can simply write
\begin{align}
I_S\left(\MA\right)\coloneqq\sup\left\{S\left({\bm\tau}\right)\,\middle|\,{\bm\sigma}\xrightarrow{\rm LF_1}{\bm\tau}\right\}
\end{align}
for {\em any} state assemblage ${\bm\sigma}$ satisfying ${\bf B}^{({\bm\sigma})} = \MA$.

\section{Proof of Theorem~\ref{Result:duality-noise-model}}\label{App:Proof-Result:duality-noise-model}
\begin{proof}
For every given unitary $U$ and state $\eta$ with the condition \mbox{${\rm supp}(\eta) \subseteq {\rm supp}(U\rho_{\bm\sigma}U^\dagger)$,} we can write ${\rm SR}_{\mathcal{N}_S}\left(\sqrt{\eta}U{\bf B}^{({\bm\sigma})}U^\dagger\sqrt{\eta}\right)$
as the following optimisation:
\begin{equation}
\begin{split}
\min\quad&t\\
{\rm s.t.}\quad& t\ge0, \frac{\sqrt{\eta}U{\bf B}^{({\bm\sigma})}U^\dagger\sqrt{\eta}+t{\bm\omega}}{1+t}\in{\bf LHS},\\
&{\bm\omega}\in\mathcal{N}_S\left(\sqrt{\eta}U{\bf B}^{({\bm\sigma})}U^\dagger\sqrt{\eta}\right).
\end{split}
\end{equation}
Using Eq.~\eqref{Eq:SEO-ordering-induced} to make the minimisation range smaller, we obtain the following upper bound:
\begin{equation}
\begin{split}
\min\quad&t\\
{\rm s.t.}\quad& t\ge0, \frac{\sqrt{\eta}U{\bf B}^{({\bm\sigma})}U^\dagger\sqrt{\eta}+t{\bm\omega}}{1+t}\in{\bf LHS},\\
&{\bm\omega}\in \sqrt{\eta}U\mathcal{N}_I\left({\bf B}^{({\bm\sigma})}\right)U^\dagger\sqrt{\eta},\\&
\end{split}
\end{equation}
This can be rewritten as
\begin{equation}\label{Eq:Comp-Proof-Th4-001}
\begin{split}
\min\quad&t\\
{\rm s.t.}\quad& t\ge0, \sqrt{\eta}U\left(\frac{{\bf B}^{({\bm\sigma})}+t{\bf W}}{1+t}\right)U^\dagger\sqrt{\eta}\in{\bf LHS},\\
&{\bf W}\in \mathcal{N}_I\left({\bf B}^{({\bm\sigma})}\right).
\end{split}
\end{equation}
Whenever ${\rm supp}(\eta)\subseteq{\rm supp}(U\id_{\bf M}U^\dagger)$, ${\bf M}\in{\bf JM}$ implies that \mbox{$\sqrt{\eta}U{\bf M}U^\dagger\sqrt{\eta}\in{\bf LHS}$}. 
Hence, Eq.~\eqref{Eq:Comp-Proof-Th4-001} is upper bounded by 
\begin{equation}
\begin{split}
\min\quad&t\\
{\rm s.t.}\quad& t\ge0, \frac{{\bf B}^{({\bm\sigma})}+t{\bf W}}{1+t}\in{\bf JM}, {\bf W}\in \mathcal{N}_I\left({\bf B}^{({\bm\sigma})}\right),\\&
\end{split}
\end{equation}
which is exactly ${\rm IR}_{\mathcal{N}_I}\left({\bf B}^{({\bm\sigma})}\right)$.
This means that ${\rm SR}_{\mathcal{N}_S}\left(\sqrt{\eta}U{\bf B}^{({\bm\sigma})}U^\dagger\sqrt{\eta}\right)\le{\rm IR}_{\mathcal{N}_I}\left({\bf B}^{({\bm\sigma})}\right)$ for every $\eta$ and unitary $U$ with ${\rm supp}(\eta) \subseteq {\rm supp}(U\rho_{\bm\sigma}U^\dagger)$.
Maximising over all such $\eta,U$ and using Theorem~\ref{Result:LF1=SEO-ordering}, we obtain
\begin{align}\label{Eq:LowerBound}
\sup\left\{{\rm SR}_{\mathcal{N}_S}\left({\bm\tau}\right)\,\middle|\,{\bm\sigma}\xrightarrow{\rm LF_1}{\bm\tau}\right\}
\le{\rm IR}_{\mathcal{N}_I}\left({\bf B}^{({\bm\sigma})}\right).
\end{align}
This concludes the proof.
\end{proof}

\bibliography{Ref.bib}

%apsrev4-2.bst 2019-01-14 (MD) hand-edited version of apsrev4-1.bst
%Control: key (0)
%Control: author (8) initials jnrlst
%Control: editor formatted (1) identically to author
%Control: production of article title (0) allowed
%Control: page (0) single
%Control: year (1) truncated
%Control: production of eprint (0) enabled
\begin{thebibliography}{66}%
\makeatletter
\providecommand \@ifxundefined [1]{%
 \@ifx{#1\undefined}
}%
\providecommand \@ifnum [1]{%
 \ifnum #1\expandafter \@firstoftwo
 \else \expandafter \@secondoftwo
 \fi
}%
\providecommand \@ifx [1]{%
 \ifx #1\expandafter \@firstoftwo
 \else \expandafter \@secondoftwo
 \fi
}%
\providecommand \natexlab [1]{#1}%
\providecommand \enquote  [1]{``#1''}%
\providecommand \bibnamefont  [1]{#1}%
\providecommand \bibfnamefont [1]{#1}%
\providecommand \citenamefont [1]{#1}%
\providecommand \href@noop [0]{\@secondoftwo}%
\providecommand \href [0]{\begingroup \@sanitize@url \@href}%
\providecommand \@href[1]{\@@startlink{#1}\@@href}%
\providecommand \@@href[1]{\endgroup#1\@@endlink}%
\providecommand \@sanitize@url [0]{\catcode `\\12\catcode `\$12\catcode
  `\&12\catcode `\#12\catcode `\^12\catcode `\_12\catcode `\%12\relax}%
\providecommand \@@startlink[1]{}%
\providecommand \@@endlink[0]{}%
\providecommand \url  [0]{\begingroup\@sanitize@url \@url }%
\providecommand \@url [1]{\endgroup\@href {#1}{\urlprefix }}%
\providecommand \urlprefix  [0]{URL }%
\providecommand \Eprint [0]{\href }%
\providecommand \doibase [0]{https://doi.org/}%
\providecommand \selectlanguage [0]{\@gobble}%
\providecommand \bibinfo  [0]{\@secondoftwo}%
\providecommand \bibfield  [0]{\@secondoftwo}%
\providecommand \translation [1]{[#1]}%
\providecommand \BibitemOpen [0]{}%
\providecommand \bibitemStop [0]{}%
\providecommand \bibitemNoStop [0]{.\EOS\space}%
\providecommand \EOS [0]{\spacefactor3000\relax}%
\providecommand \BibitemShut  [1]{\csname bibitem#1\endcsname}%
\let\auto@bib@innerbib\@empty
%</preamble>
\bibitem [{\citenamefont {Chitambar}\ and\ \citenamefont
  {Gour}(2019)}]{ChitambarRMP2019}%
  \BibitemOpen
  \bibfield  {author} {\bibinfo {author} {\bibfnamefont {E.}~\bibnamefont
  {Chitambar}}\ and\ \bibinfo {author} {\bibfnamefont {G.}~\bibnamefont
  {Gour}},\ }\bibfield  {title} {\bibinfo {title} {Quantum resource theories},\
  }\href {https://doi.org/10.1103/RevModPhys.91.025001} {\bibfield  {journal}
  {\bibinfo  {journal} {Rev. Mod. Phys.}\ }\textbf {\bibinfo {volume} {91}},\
  \bibinfo {pages} {025001} (\bibinfo {year} {2019})}\BibitemShut {NoStop}%
\bibitem [{\citenamefont {Bennett}\ \emph {et~al.}(1993)\citenamefont
  {Bennett}, \citenamefont {Brassard}, \citenamefont {Cr\'epeau}, \citenamefont
  {Jozsa}, \citenamefont {Peres},\ and\ \citenamefont {Wootters}}]{Bennett93}%
  \BibitemOpen
  \bibfield  {author} {\bibinfo {author} {\bibfnamefont {C.~H.}\ \bibnamefont
  {Bennett}}, \bibinfo {author} {\bibfnamefont {G.}~\bibnamefont {Brassard}},
  \bibinfo {author} {\bibfnamefont {C.}~\bibnamefont {Cr\'epeau}}, \bibinfo
  {author} {\bibfnamefont {R.}~\bibnamefont {Jozsa}}, \bibinfo {author}
  {\bibfnamefont {A.}~\bibnamefont {Peres}},\ and\ \bibinfo {author}
  {\bibfnamefont {W.~K.}\ \bibnamefont {Wootters}},\ }\bibfield  {title}
  {\bibinfo {title} {Teleporting an unknown quantum state via dual classical
  and {E}instein-{P}odolsky-{R}osen channels},\ }\href
  {https://doi.org/10.1103/PhysRevLett.70.1895} {\bibfield  {journal} {\bibinfo
   {journal} {Phys. Rev. Lett.}\ }\textbf {\bibinfo {volume} {70}},\ \bibinfo
  {pages} {1895} (\bibinfo {year} {1993})}\BibitemShut {NoStop}%
\bibitem [{\citenamefont {Bennett}\ and\ \citenamefont
  {Wiesner}(1992)}]{Bennett92}%
  \BibitemOpen
  \bibfield  {author} {\bibinfo {author} {\bibfnamefont {C.~H.}\ \bibnamefont
  {Bennett}}\ and\ \bibinfo {author} {\bibfnamefont {S.~J.}\ \bibnamefont
  {Wiesner}},\ }\bibfield  {title} {\bibinfo {title} {Communication via one-
  and two-particle operators on {E}instein-{P}odolsky-{R}osen states},\ }\href
  {https://doi.org/10.1103/PhysRevLett.69.2881} {\bibfield  {journal} {\bibinfo
   {journal} {Phys. Rev. Lett.}\ }\textbf {\bibinfo {volume} {69}},\ \bibinfo
  {pages} {2881} (\bibinfo {year} {1992})}\BibitemShut {NoStop}%
\bibitem [{\citenamefont {Wilde}(2017)}]{wilde_2017}%
  \BibitemOpen
  \bibfield  {author} {\bibinfo {author} {\bibfnamefont {M.~M.}\ \bibnamefont
  {Wilde}},\ }\href {https://doi.org/10.1017/9781316809976} {\emph {\bibinfo
  {title} {Quantum Information Theory}}},\ \bibinfo {edition} {2nd}\ ed.\
  (\bibinfo  {publisher} {Cambridge University Press},\ \bibinfo {year}
  {2017})\BibitemShut {NoStop}%
\bibitem [{\citenamefont {Bennett}\ \emph {et~al.}(1996)\citenamefont
  {Bennett}, \citenamefont {DiVincenzo}, \citenamefont {Smolin},\ and\
  \citenamefont {Wootters}}]{Bennett1996}%
  \BibitemOpen
  \bibfield  {author} {\bibinfo {author} {\bibfnamefont {C.~H.}\ \bibnamefont
  {Bennett}}, \bibinfo {author} {\bibfnamefont {D.~P.}\ \bibnamefont
  {DiVincenzo}}, \bibinfo {author} {\bibfnamefont {J.~A.}\ \bibnamefont
  {Smolin}},\ and\ \bibinfo {author} {\bibfnamefont {W.~K.}\ \bibnamefont
  {Wootters}},\ }\bibfield  {title} {\bibinfo {title} {Mixed-state entanglement
  and quantum error correction},\ }\href
  {https://doi.org/10.1103/PhysRevA.54.3824} {\bibfield  {journal} {\bibinfo
  {journal} {Phys. Rev. A}\ }\textbf {\bibinfo {volume} {54}},\ \bibinfo
  {pages} {3824} (\bibinfo {year} {1996})}\BibitemShut {NoStop}%
\bibitem [{\citenamefont {Horodecki}\ \emph {et~al.}(1999)\citenamefont
  {Horodecki}, \citenamefont {Horodecki},\ and\ \citenamefont
  {Horodecki}}]{Horodecki1999}%
  \BibitemOpen
  \bibfield  {author} {\bibinfo {author} {\bibfnamefont {M.}~\bibnamefont
  {Horodecki}}, \bibinfo {author} {\bibfnamefont {P.}~\bibnamefont
  {Horodecki}},\ and\ \bibinfo {author} {\bibfnamefont {R.}~\bibnamefont
  {Horodecki}},\ }\bibfield  {title} {\bibinfo {title} {General teleportation
  channel, singlet fraction, and quasidistillation},\ }\href
  {https://doi.org/10.1103/PhysRevA.60.1888} {\bibfield  {journal} {\bibinfo
  {journal} {Phys. Rev. A}\ }\textbf {\bibinfo {volume} {60}},\ \bibinfo
  {pages} {1888} (\bibinfo {year} {1999})}\BibitemShut {NoStop}%
\bibitem [{\citenamefont {Regula}\ and\ \citenamefont
  {Takagi}(2021{\natexlab{a}})}]{Regula2021}%
  \BibitemOpen
  \bibfield  {author} {\bibinfo {author} {\bibfnamefont {B.}~\bibnamefont
  {Regula}}\ and\ \bibinfo {author} {\bibfnamefont {R.}~\bibnamefont
  {Takagi}},\ }\bibfield  {title} {\bibinfo {title} {Fundamental limitations on
  distillation of quantum channel resources},\ }\href
  {https://doi.org/10.1038/s41467-021-24699-0} {\bibfield  {journal} {\bibinfo
  {journal} {Nat. Commun.}\ }\textbf {\bibinfo {volume} {12}},\ \bibinfo
  {pages} {4411} (\bibinfo {year} {2021}{\natexlab{a}})}\BibitemShut {NoStop}%
\bibitem [{\citenamefont {Regula}\ and\ \citenamefont
  {Takagi}(2021{\natexlab{b}})}]{Regula2021PRL}%
  \BibitemOpen
  \bibfield  {author} {\bibinfo {author} {\bibfnamefont {B.}~\bibnamefont
  {Regula}}\ and\ \bibinfo {author} {\bibfnamefont {R.}~\bibnamefont
  {Takagi}},\ }\bibfield  {title} {\bibinfo {title} {One-shot manipulation of
  dynamical quantum resources},\ }\href
  {https://doi.org/10.1103/PhysRevLett.127.060402} {\bibfield  {journal}
  {\bibinfo  {journal} {Phys. Rev. Lett.}\ }\textbf {\bibinfo {volume} {127}},\
  \bibinfo {pages} {060402} (\bibinfo {year} {2021}{\natexlab{b}})}\BibitemShut
  {NoStop}%
\bibitem [{\citenamefont {Lostaglio}(2019)}]{Lostaglio2019}%
  \BibitemOpen
  \bibfield  {author} {\bibinfo {author} {\bibfnamefont {M.}~\bibnamefont
  {Lostaglio}},\ }\bibfield  {title} {\bibinfo {title} {An introductory review
  of the resource theory approach to thermodynamics},\ }\href
  {https://doi.org/10.1088/1361-6633/ab46e5} {\bibfield  {journal} {\bibinfo
  {journal} {Rep. Prog. Phys.}\ }\textbf {\bibinfo {volume} {82}},\ \bibinfo
  {pages} {114001} (\bibinfo {year} {2019})}\BibitemShut {NoStop}%
\bibitem [{\citenamefont {Gour}\ \emph {et~al.}(2015)\citenamefont {Gour},
  \citenamefont {Müller}, \citenamefont {Narasimhachar}, \citenamefont
  {Spekkens},\ and\ \citenamefont {{Yunger Halpern}}}]{Purity-review}%
  \BibitemOpen
  \bibfield  {author} {\bibinfo {author} {\bibfnamefont {G.}~\bibnamefont
  {Gour}}, \bibinfo {author} {\bibfnamefont {M.~P.}\ \bibnamefont {Müller}},
  \bibinfo {author} {\bibfnamefont {V.}~\bibnamefont {Narasimhachar}}, \bibinfo
  {author} {\bibfnamefont {R.~W.}\ \bibnamefont {Spekkens}},\ and\ \bibinfo
  {author} {\bibfnamefont {N.}~\bibnamefont {{Yunger Halpern}}},\ }\bibfield
  {title} {\bibinfo {title} {The resource theory of informational
  nonequilibrium in thermodynamics},\ }\href
  {https://doi.org/https://doi.org/10.1016/j.physrep.2015.04.003} {\bibfield
  {journal} {\bibinfo  {journal} {Phys. Rep.}\ }\textbf {\bibinfo {volume}
  {583}},\ \bibinfo {pages} {1} (\bibinfo {year} {2015})}\BibitemShut {NoStop}%
\bibitem [{\citenamefont {Linden}\ and\ \citenamefont
  {Skrzypczyk}()}]{Paul2022}%
  \BibitemOpen
  \bibfield  {author} {\bibinfo {author} {\bibfnamefont {N.}~\bibnamefont
  {Linden}}\ and\ \bibinfo {author} {\bibfnamefont {P.}~\bibnamefont
  {Skrzypczyk}},\ }\href@noop {} {\bibinfo {title} {How to use arbitrary
  measuring devices to perform almost perfect measurements}},\ \Eprint
  {https://arxiv.org/abs/2203.02593} {arXiv:2203.02593} \BibitemShut {NoStop}%
\bibitem [{\citenamefont {Gisin}(1996)}]{Gisin1996PLA}%
  \BibitemOpen
  \bibfield  {author} {\bibinfo {author} {\bibfnamefont {N.}~\bibnamefont
  {Gisin}},\ }\bibfield  {title} {\bibinfo {title} {Hidden quantum nonlocality
  revealed by local filters},\ }\href
  {https://doi.org/https://doi.org/10.1016/S0375-9601(96)80001-6} {\bibfield
  {journal} {\bibinfo  {journal} {Phys. Lett. A}\ }\textbf {\bibinfo {volume}
  {210}},\ \bibinfo {pages} {151} (\bibinfo {year} {1996})}\BibitemShut
  {NoStop}%
\bibitem [{\citenamefont {Forster}\ \emph {et~al.}(2009)\citenamefont
  {Forster}, \citenamefont {Winkler},\ and\ \citenamefont
  {Wolf}}]{Forster2009PRL}%
  \BibitemOpen
  \bibfield  {author} {\bibinfo {author} {\bibfnamefont {M.}~\bibnamefont
  {Forster}}, \bibinfo {author} {\bibfnamefont {S.}~\bibnamefont {Winkler}},\
  and\ \bibinfo {author} {\bibfnamefont {S.}~\bibnamefont {Wolf}},\ }\bibfield
  {title} {\bibinfo {title} {Distilling nonlocality},\ }\href
  {https://doi.org/10.1103/PhysRevLett.102.120401} {\bibfield  {journal}
  {\bibinfo  {journal} {Phys. Rev. Lett.}\ }\textbf {\bibinfo {volume} {102}},\
  \bibinfo {pages} {120401} (\bibinfo {year} {2009})}\BibitemShut {NoStop}%
\bibitem [{\citenamefont {Naik}\ \emph {et~al.}(2023)\citenamefont {Naik},
  \citenamefont {Sidhardh}, \citenamefont {Sen}, \citenamefont {Roy},
  \citenamefont {Rai},\ and\ \citenamefont {Banik}}]{Naik2023PRL}%
  \BibitemOpen
  \bibfield  {author} {\bibinfo {author} {\bibfnamefont {S.~G.}\ \bibnamefont
  {Naik}}, \bibinfo {author} {\bibfnamefont {G.~L.}\ \bibnamefont {Sidhardh}},
  \bibinfo {author} {\bibfnamefont {S.}~\bibnamefont {Sen}}, \bibinfo {author}
  {\bibfnamefont {A.}~\bibnamefont {Roy}}, \bibinfo {author} {\bibfnamefont
  {A.}~\bibnamefont {Rai}},\ and\ \bibinfo {author} {\bibfnamefont
  {M.}~\bibnamefont {Banik}},\ }\bibfield  {title} {\bibinfo {title}
  {Distilling nonlocality in quantum correlations},\ }\href
  {https://doi.org/10.1103/PhysRevLett.130.220201} {\bibfield  {journal}
  {\bibinfo  {journal} {Phys. Rev. Lett.}\ }\textbf {\bibinfo {volume} {130}},\
  \bibinfo {pages} {220201} (\bibinfo {year} {2023})}\BibitemShut {NoStop}%
\bibitem [{\citenamefont {Nery}\ \emph {et~al.}(2020)\citenamefont {Nery},
  \citenamefont {Taddei}, \citenamefont {Sahium}, \citenamefont {Walborn},
  \citenamefont {Aolita},\ and\ \citenamefont {Aguilar}}]{Nery2020}%
  \BibitemOpen
  \bibfield  {author} {\bibinfo {author} {\bibfnamefont {R.~V.}\ \bibnamefont
  {Nery}}, \bibinfo {author} {\bibfnamefont {M.~M.}\ \bibnamefont {Taddei}},
  \bibinfo {author} {\bibfnamefont {P.}~\bibnamefont {Sahium}}, \bibinfo
  {author} {\bibfnamefont {S.~P.}\ \bibnamefont {Walborn}}, \bibinfo {author}
  {\bibfnamefont {L.}~\bibnamefont {Aolita}},\ and\ \bibinfo {author}
  {\bibfnamefont {G.~H.}\ \bibnamefont {Aguilar}},\ }\bibfield  {title}
  {\bibinfo {title} {Distillation of quantum steering},\ }\href
  {https://doi.org/10.1103/PhysRevLett.124.120402} {\bibfield  {journal}
  {\bibinfo  {journal} {Phys. Rev. Lett.}\ }\textbf {\bibinfo {volume} {124}},\
  \bibinfo {pages} {120402} (\bibinfo {year} {2020})}\BibitemShut {NoStop}%
\bibitem [{\citenamefont {Liu}\ \emph {et~al.}(2022)\citenamefont {Liu},
  \citenamefont {Zheng}, \citenamefont {Kang}, \citenamefont {Han},
  \citenamefont {Wang}, \citenamefont {Zhang}, \citenamefont {Su},\ and\
  \citenamefont {Peng}}]{Liu2022}%
  \BibitemOpen
  \bibfield  {author} {\bibinfo {author} {\bibfnamefont {Y.}~\bibnamefont
  {Liu}}, \bibinfo {author} {\bibfnamefont {K.}~\bibnamefont {Zheng}}, \bibinfo
  {author} {\bibfnamefont {H.}~\bibnamefont {Kang}}, \bibinfo {author}
  {\bibfnamefont {D.}~\bibnamefont {Han}}, \bibinfo {author} {\bibfnamefont
  {M.}~\bibnamefont {Wang}}, \bibinfo {author} {\bibfnamefont {L.}~\bibnamefont
  {Zhang}}, \bibinfo {author} {\bibfnamefont {X.}~\bibnamefont {Su}},\ and\
  \bibinfo {author} {\bibfnamefont {K.}~\bibnamefont {Peng}},\ }\bibfield
  {title} {\bibinfo {title} {Distillation of gaussian einstein-podolsky-rosen
  steering with noiseless linear amplification},\ }\href
  {https://doi.org/10.1038/s41534-022-00549-9} {\bibfield  {journal} {\bibinfo
  {journal} {npj Quantum Inf.}\ }\textbf {\bibinfo {volume} {8}},\ \bibinfo
  {pages} {38} (\bibinfo {year} {2022})}\BibitemShut {NoStop}%
\bibitem [{\citenamefont {Ku}\ \emph {et~al.}(2022{\natexlab{a}})\citenamefont
  {Ku}, \citenamefont {Hsieh}, \citenamefont {Chen}, \citenamefont {Chen},\
  and\ \citenamefont {Budroni}}]{Ku2022NC}%
  \BibitemOpen
  \bibfield  {author} {\bibinfo {author} {\bibfnamefont {H.-Y.}\ \bibnamefont
  {Ku}}, \bibinfo {author} {\bibfnamefont {C.-Y.}\ \bibnamefont {Hsieh}},
  \bibinfo {author} {\bibfnamefont {S.-L.}\ \bibnamefont {Chen}}, \bibinfo
  {author} {\bibfnamefont {Y.-N.}\ \bibnamefont {Chen}},\ and\ \bibinfo
  {author} {\bibfnamefont {C.}~\bibnamefont {Budroni}},\ }\bibfield  {title}
  {\bibinfo {title} {Complete classification of steerability under local
  filters and its relation with measurement incompatibility},\ }\href
  {https://doi.org/10.1038/s41467-022-32466-y} {\bibfield  {journal} {\bibinfo
  {journal} {Nat. Commun.}\ }\textbf {\bibinfo {volume} {13}},\ \bibinfo
  {pages} {4973} (\bibinfo {year} {2022}{\natexlab{a}})}\BibitemShut {NoStop}%
\bibitem [{\citenamefont {Ku}\ \emph {et~al.}()\citenamefont {Ku},
  \citenamefont {Hsieh},\ and\ \citenamefont {Budroni}}]{Ku2023}%
  \BibitemOpen
  \bibfield  {author} {\bibinfo {author} {\bibfnamefont {H.-Y.}\ \bibnamefont
  {Ku}}, \bibinfo {author} {\bibfnamefont {C.-Y.}\ \bibnamefont {Hsieh}},\ and\
  \bibinfo {author} {\bibfnamefont {C.}~\bibnamefont {Budroni}},\ }\href@noop
  {} {\bibinfo {title} {Measurement incompatibility cannot be stochastically
  distilled}},\ \Eprint {https://arxiv.org/abs/2308.02252} {arXiv:2308.02252}
  \BibitemShut {NoStop}%
\bibitem [{\citenamefont {Chen}\ \emph {et~al.}(2016)\citenamefont {Chen},
  \citenamefont {Budroni}, \citenamefont {Liang},\ and\ \citenamefont
  {Chen}}]{ShinLiang2016PRL}%
  \BibitemOpen
  \bibfield  {author} {\bibinfo {author} {\bibfnamefont {S.-L.}\ \bibnamefont
  {Chen}}, \bibinfo {author} {\bibfnamefont {C.}~\bibnamefont {Budroni}},
  \bibinfo {author} {\bibfnamefont {Y.-C.}\ \bibnamefont {Liang}},\ and\
  \bibinfo {author} {\bibfnamefont {Y.-N.}\ \bibnamefont {Chen}},\ }\bibfield
  {title} {\bibinfo {title} {Natural framework for device-independent
  quantification of quantum steerability, measurement incompatibility, and
  self-testing},\ }\href {https://doi.org/10.1103/PhysRevLett.116.240401}
  {\bibfield  {journal} {\bibinfo  {journal} {Phys. Rev. Lett.}\ }\textbf
  {\bibinfo {volume} {116}},\ \bibinfo {pages} {240401} (\bibinfo {year}
  {2016})}\BibitemShut {NoStop}%
\bibitem [{\citenamefont {Skrzypczyk}\ and\ \citenamefont
  {Cavalcanti}(2018{\natexlab{a}})}]{Skrzypczyk2018}%
  \BibitemOpen
  \bibfield  {author} {\bibinfo {author} {\bibfnamefont {P.}~\bibnamefont
  {Skrzypczyk}}\ and\ \bibinfo {author} {\bibfnamefont {D.}~\bibnamefont
  {Cavalcanti}},\ }\bibfield  {title} {\bibinfo {title} {Maximal randomness
  generation from steering inequality violations using qudits},\ }\href
  {https://doi.org/10.1103/PhysRevLett.120.260401} {\bibfield  {journal}
  {\bibinfo  {journal} {Phys. Rev. Lett.}\ }\textbf {\bibinfo {volume} {120}},\
  \bibinfo {pages} {260401} (\bibinfo {year} {2018}{\natexlab{a}})}\BibitemShut
  {NoStop}%
\bibitem [{\citenamefont {Chen}\ \emph {et~al.}(2021)\citenamefont {Chen},
  \citenamefont {Ku}, \citenamefont {Zhou}, \citenamefont {Tura},\ and\
  \citenamefont {Chen}}]{Chen2021robustselftestingof}%
  \BibitemOpen
  \bibfield  {author} {\bibinfo {author} {\bibfnamefont {S.-L.}\ \bibnamefont
  {Chen}}, \bibinfo {author} {\bibfnamefont {H.-Y.}\ \bibnamefont {Ku}},
  \bibinfo {author} {\bibfnamefont {W.}~\bibnamefont {Zhou}}, \bibinfo {author}
  {\bibfnamefont {J.}~\bibnamefont {Tura}},\ and\ \bibinfo {author}
  {\bibfnamefont {Y.-N.}\ \bibnamefont {Chen}},\ }\bibfield  {title} {\bibinfo
  {title} {Robust self-testing of steerable quantum assemblages and its
  applications on device-independent quantum certification},\ }\href
  {https://doi.org/10.22331/q-2021-09-28-552} {\bibfield  {journal} {\bibinfo
  {journal} {{Quantum}}\ }\textbf {\bibinfo {volume} {5}},\ \bibinfo {pages}
  {552} (\bibinfo {year} {2021})}\BibitemShut {NoStop}%
\bibitem [{\citenamefont {Branciard}\ \emph {et~al.}(2012)\citenamefont
  {Branciard}, \citenamefont {Cavalcanti}, \citenamefont {Walborn},
  \citenamefont {Scarani},\ and\ \citenamefont {Wiseman}}]{Branciard2012}%
  \BibitemOpen
  \bibfield  {author} {\bibinfo {author} {\bibfnamefont {C.}~\bibnamefont
  {Branciard}}, \bibinfo {author} {\bibfnamefont {E.~G.}\ \bibnamefont
  {Cavalcanti}}, \bibinfo {author} {\bibfnamefont {S.~P.}\ \bibnamefont
  {Walborn}}, \bibinfo {author} {\bibfnamefont {V.}~\bibnamefont {Scarani}},\
  and\ \bibinfo {author} {\bibfnamefont {H.~M.}\ \bibnamefont {Wiseman}},\
  }\bibfield  {title} {\bibinfo {title} {One-sided device-independent quantum
  key distribution: Security, feasibility, and the connection with steering},\
  }\href {https://link.aps.org/doi/10.1103/PhysRevA.85.010301} {\bibfield
  {journal} {\bibinfo  {journal} {Phys. Rev. A}\ }\textbf {\bibinfo {volume}
  {85}},\ \bibinfo {pages} {010301} (\bibinfo {year} {2012})}\BibitemShut
  {NoStop}%
\bibitem [{\citenamefont {Piani}\ and\ \citenamefont
  {Watrous}(2015)}]{Piani2015}%
  \BibitemOpen
  \bibfield  {author} {\bibinfo {author} {\bibfnamefont {M.}~\bibnamefont
  {Piani}}\ and\ \bibinfo {author} {\bibfnamefont {J.}~\bibnamefont
  {Watrous}},\ }\bibfield  {title} {\bibinfo {title} {Necessary and sufficient
  quantum information characterization of {E}instein-{P}odolsky-{R}osen
  steering},\ }\href {https://doi.org/10.1103/PhysRevLett.114.060404}
  {\bibfield  {journal} {\bibinfo  {journal} {Phys. Rev. Lett.}\ }\textbf
  {\bibinfo {volume} {114}},\ \bibinfo {pages} {060404} (\bibinfo {year}
  {2015})}\BibitemShut {NoStop}%
\bibitem [{\citenamefont {Zhao}\ \emph {et~al.}(2020)\citenamefont {Zhao},
  \citenamefont {Ku}, \citenamefont {Chen}, \citenamefont {Chen}, \citenamefont
  {Nori}, \citenamefont {Xiang}, \citenamefont {Li}, \citenamefont {Guo},\ and\
  \citenamefont {Chen}}]{Zhao2020}%
  \BibitemOpen
  \bibfield  {author} {\bibinfo {author} {\bibfnamefont {Y.-Y.}\ \bibnamefont
  {Zhao}}, \bibinfo {author} {\bibfnamefont {H.-Y.}\ \bibnamefont {Ku}},
  \bibinfo {author} {\bibfnamefont {S.-L.}\ \bibnamefont {Chen}}, \bibinfo
  {author} {\bibfnamefont {H.-B.}\ \bibnamefont {Chen}}, \bibinfo {author}
  {\bibfnamefont {F.}~\bibnamefont {Nori}}, \bibinfo {author} {\bibfnamefont
  {G.-Y.}\ \bibnamefont {Xiang}}, \bibinfo {author} {\bibfnamefont {C.-F.}\
  \bibnamefont {Li}}, \bibinfo {author} {\bibfnamefont {G.-C.}\ \bibnamefont
  {Guo}},\ and\ \bibinfo {author} {\bibfnamefont {Y.-N.}\ \bibnamefont
  {Chen}},\ }\bibfield  {title} {\bibinfo {title} {Experimental demonstration
  of measurement-device-independent measure of quantum steering},\ }\href
  {https://doi.org/10.1038/s41534-020-00307-9} {\bibfield  {journal} {\bibinfo
  {journal} {npj Quantum Inf.}\ }\textbf {\bibinfo {volume} {6}},\ \bibinfo
  {pages} {77} (\bibinfo {year} {2020})}\BibitemShut {NoStop}%
\bibitem [{\citenamefont {Ku}\ \emph {et~al.}(2022{\natexlab{b}})\citenamefont
  {Ku}, \citenamefont {Kadlec}, \citenamefont {\ifmmode~\check{C}\else
  \v{C}\fi{}ernoch}, \citenamefont {Quintino}, \citenamefont {Zhou},
  \citenamefont {Lemr}, \citenamefont {Lambert}, \citenamefont {Miranowicz},
  \citenamefont {Chen}, \citenamefont {Nori},\ and\ \citenamefont
  {Chen}}]{Ku2022PRXQ}%
  \BibitemOpen
  \bibfield  {author} {\bibinfo {author} {\bibfnamefont {H.-Y.}\ \bibnamefont
  {Ku}}, \bibinfo {author} {\bibfnamefont {J.}~\bibnamefont {Kadlec}}, \bibinfo
  {author} {\bibfnamefont {A.}~\bibnamefont {\ifmmode~\check{C}\else
  \v{C}\fi{}ernoch}}, \bibinfo {author} {\bibfnamefont {M.~T.}\ \bibnamefont
  {Quintino}}, \bibinfo {author} {\bibfnamefont {W.}~\bibnamefont {Zhou}},
  \bibinfo {author} {\bibfnamefont {K.}~\bibnamefont {Lemr}}, \bibinfo {author}
  {\bibfnamefont {N.}~\bibnamefont {Lambert}}, \bibinfo {author} {\bibfnamefont
  {A.}~\bibnamefont {Miranowicz}}, \bibinfo {author} {\bibfnamefont {S.-L.}\
  \bibnamefont {Chen}}, \bibinfo {author} {\bibfnamefont {F.}~\bibnamefont
  {Nori}},\ and\ \bibinfo {author} {\bibfnamefont {Y.-N.}\ \bibnamefont
  {Chen}},\ }\bibfield  {title} {\bibinfo {title} {Quantifying quantumness of
  channels without entanglement},\ }\href
  {https://doi.org/10.1103/PRXQuantum.3.020338} {\bibfield  {journal} {\bibinfo
   {journal} {PRX Quantum}\ }\textbf {\bibinfo {volume} {3}},\ \bibinfo {pages}
  {020338} (\bibinfo {year} {2022}{\natexlab{b}})}\BibitemShut {NoStop}%
\bibitem [{\citenamefont {G\"uhne}\ \emph {et~al.}(2023)\citenamefont
  {G\"uhne}, \citenamefont {Haapasalo}, \citenamefont {Kraft}, \citenamefont
  {Pellonp\"a\"a},\ and\ \citenamefont {Uola}}]{Otfried2021Rev}%
  \BibitemOpen
  \bibfield  {author} {\bibinfo {author} {\bibfnamefont {O.}~\bibnamefont
  {G\"uhne}}, \bibinfo {author} {\bibfnamefont {E.}~\bibnamefont {Haapasalo}},
  \bibinfo {author} {\bibfnamefont {T.}~\bibnamefont {Kraft}}, \bibinfo
  {author} {\bibfnamefont {J.-P.}\ \bibnamefont {Pellonp\"a\"a}},\ and\
  \bibinfo {author} {\bibfnamefont {R.}~\bibnamefont {Uola}},\ }\bibfield
  {title} {\bibinfo {title} {Colloquium: Incompatible measurements in quantum
  information science},\ }\href {https://doi.org/10.1103/RevModPhys.95.011003}
  {\bibfield  {journal} {\bibinfo  {journal} {Rev. Mod. Phys.}\ }\textbf
  {\bibinfo {volume} {95}},\ \bibinfo {pages} {011003} (\bibinfo {year}
  {2023})}\BibitemShut {NoStop}%
\bibitem [{\citenamefont {Wiseman}\ \emph {et~al.}(2007)\citenamefont
  {Wiseman}, \citenamefont {Jones},\ and\ \citenamefont
  {Doherty}}]{Wiseman2007PRL}%
  \BibitemOpen
  \bibfield  {author} {\bibinfo {author} {\bibfnamefont {H.~M.}\ \bibnamefont
  {Wiseman}}, \bibinfo {author} {\bibfnamefont {S.~J.}\ \bibnamefont {Jones}},\
  and\ \bibinfo {author} {\bibfnamefont {A.~C.}\ \bibnamefont {Doherty}},\
  }\bibfield  {title} {\bibinfo {title} {Steering, entanglement, nonlocality,
  and the {E}instein-{P}odolsky-{R}osen paradox},\ }\href
  {https://doi.org/10.1103/PhysRevLett.98.140402} {\bibfield  {journal}
  {\bibinfo  {journal} {Phys. Rev. Lett.}\ }\textbf {\bibinfo {volume} {98}},\
  \bibinfo {pages} {140402} (\bibinfo {year} {2007})}\BibitemShut {NoStop}%
\bibitem [{\citenamefont {Cavalcanti}\ and\ \citenamefont
  {Skrzypczyk}(2016{\natexlab{a}})}]{Cavalcanti2016}%
  \BibitemOpen
  \bibfield  {author} {\bibinfo {author} {\bibfnamefont {D.}~\bibnamefont
  {Cavalcanti}}\ and\ \bibinfo {author} {\bibfnamefont {P.}~\bibnamefont
  {Skrzypczyk}},\ }\bibfield  {title} {\bibinfo {title} {Quantum steering: a
  review with focus on semidefinite programming},\ }\href
  {https://doi.org/10.1088/1361-6633/80/2/024001} {\bibfield  {journal}
  {\bibinfo  {journal} {Rep. Prog. Phys.}\ }\textbf {\bibinfo {volume} {80}},\
  \bibinfo {pages} {024001} (\bibinfo {year} {2016}{\natexlab{a}})}\BibitemShut
  {NoStop}%
\bibitem [{\citenamefont {Uola}\ \emph {et~al.}(2020)\citenamefont {Uola},
  \citenamefont {Costa}, \citenamefont {Nguyen},\ and\ \citenamefont
  {G\"uhne}}]{UolaRMP2020}%
  \BibitemOpen
  \bibfield  {author} {\bibinfo {author} {\bibfnamefont {R.}~\bibnamefont
  {Uola}}, \bibinfo {author} {\bibfnamefont {A.~C.~S.}\ \bibnamefont {Costa}},
  \bibinfo {author} {\bibfnamefont {H.~C.}\ \bibnamefont {Nguyen}},\ and\
  \bibinfo {author} {\bibfnamefont {O.}~\bibnamefont {G\"uhne}},\ }\bibfield
  {title} {\bibinfo {title} {Quantum steering},\ }\href
  {https://doi.org/10.1103/RevModPhys.92.015001} {\bibfield  {journal}
  {\bibinfo  {journal} {Rev. Mod. Phys.}\ }\textbf {\bibinfo {volume} {92}},\
  \bibinfo {pages} {015001} (\bibinfo {year} {2020})}\BibitemShut {NoStop}%
\bibitem [{\citenamefont {Xiang}\ \emph {et~al.}(2022)\citenamefont {Xiang},
  \citenamefont {Cheng}, \citenamefont {Gong}, \citenamefont {Ficek},\ and\
  \citenamefont {He}}]{XiangPRXQ2022}%
  \BibitemOpen
  \bibfield  {author} {\bibinfo {author} {\bibfnamefont {Y.}~\bibnamefont
  {Xiang}}, \bibinfo {author} {\bibfnamefont {S.}~\bibnamefont {Cheng}},
  \bibinfo {author} {\bibfnamefont {Q.}~\bibnamefont {Gong}}, \bibinfo {author}
  {\bibfnamefont {Z.}~\bibnamefont {Ficek}},\ and\ \bibinfo {author}
  {\bibfnamefont {Q.}~\bibnamefont {He}},\ }\bibfield  {title} {\bibinfo
  {title} {Quantum steering: Practical challenges and future directions},\
  }\href {https://doi.org/10.1103/PRXQuantum.3.030102} {\bibfield  {journal}
  {\bibinfo  {journal} {PRX Quantum}\ }\textbf {\bibinfo {volume} {3}},\
  \bibinfo {pages} {030102} (\bibinfo {year} {2022})}\BibitemShut {NoStop}%
\bibitem [{\citenamefont {Nielsen}\ and\ \citenamefont
  {Chuang}(2010)}]{QIC-book}%
  \BibitemOpen
  \bibfield  {author} {\bibinfo {author} {\bibfnamefont {M.~A.}\ \bibnamefont
  {Nielsen}}\ and\ \bibinfo {author} {\bibfnamefont {I.~L.}\ \bibnamefont
  {Chuang}},\ }\href@noop {} {\emph {\bibinfo {title} {Quantum Computation and
  Quantum Information}}},\ \bibinfo {edition} {10th}\ ed.\ (\bibinfo
  {publisher} {Cambridge University Press},\ \bibinfo {year}
  {2010})\BibitemShut {NoStop}%
\bibitem [{\citenamefont {Busch}\ \emph {et~al.}(2014)\citenamefont {Busch},
  \citenamefont {Lahti},\ and\ \citenamefont {Werner}}]{Busch2014RMP}%
  \BibitemOpen
  \bibfield  {author} {\bibinfo {author} {\bibfnamefont {P.}~\bibnamefont
  {Busch}}, \bibinfo {author} {\bibfnamefont {P.}~\bibnamefont {Lahti}},\ and\
  \bibinfo {author} {\bibfnamefont {R.~F.}\ \bibnamefont {Werner}},\ }\bibfield
   {title} {\bibinfo {title} {Colloquium: Quantum root-mean-square error and
  measurement uncertainty relations},\ }\href
  {https://doi.org/10.1103/RevModPhys.86.1261} {\bibfield  {journal} {\bibinfo
  {journal} {Rev. Mod. Phys.}\ }\textbf {\bibinfo {volume} {86}},\ \bibinfo
  {pages} {1261} (\bibinfo {year} {2014})}\BibitemShut {NoStop}%
\bibitem [{\citenamefont {Wolf}\ \emph {et~al.}(2009)\citenamefont {Wolf},
  \citenamefont {Perez-Garcia},\ and\ \citenamefont {Fernandez}}]{Wolf2009PRL}%
  \BibitemOpen
  \bibfield  {author} {\bibinfo {author} {\bibfnamefont {M.~M.}\ \bibnamefont
  {Wolf}}, \bibinfo {author} {\bibfnamefont {D.}~\bibnamefont {Perez-Garcia}},\
  and\ \bibinfo {author} {\bibfnamefont {C.}~\bibnamefont {Fernandez}},\
  }\bibfield  {title} {\bibinfo {title} {Measurements incompatible in quantum
  theory cannot be measured jointly in any other no-signaling theory},\ }\href
  {https://doi.org/10.1103/PhysRevLett.103.230402} {\bibfield  {journal}
  {\bibinfo  {journal} {Phys. Rev. Lett.}\ }\textbf {\bibinfo {volume} {103}},\
  \bibinfo {pages} {230402} (\bibinfo {year} {2009})}\BibitemShut {NoStop}%
\bibitem [{\citenamefont {Quintino}\ \emph {et~al.}(2019)\citenamefont
  {Quintino}, \citenamefont {Budroni}, \citenamefont {Woodhead}, \citenamefont
  {Cabello},\ and\ \citenamefont {Cavalcanti}}]{QuintinoPRL2019}%
  \BibitemOpen
  \bibfield  {author} {\bibinfo {author} {\bibfnamefont {M.~T.}\ \bibnamefont
  {Quintino}}, \bibinfo {author} {\bibfnamefont {C.}~\bibnamefont {Budroni}},
  \bibinfo {author} {\bibfnamefont {E.}~\bibnamefont {Woodhead}}, \bibinfo
  {author} {\bibfnamefont {A.}~\bibnamefont {Cabello}},\ and\ \bibinfo {author}
  {\bibfnamefont {D.}~\bibnamefont {Cavalcanti}},\ }\bibfield  {title}
  {\bibinfo {title} {Device-independent tests of structures of measurement
  incompatibility},\ }\href {https://doi.org/10.1103/PhysRevLett.123.180401}
  {\bibfield  {journal} {\bibinfo  {journal} {Phys. Rev. Lett.}\ }\textbf
  {\bibinfo {volume} {123}},\ \bibinfo {pages} {180401} (\bibinfo {year}
  {2019})}\BibitemShut {NoStop}%
\bibitem [{\citenamefont {Zhao}\ \emph {et~al.}(2023)\citenamefont {Zhao},
  \citenamefont {Zhang}, \citenamefont {Cheng}, \citenamefont {Li},
  \citenamefont {Guo}, \citenamefont {Liu}, \citenamefont {Ku}, \citenamefont
  {Chen}, \citenamefont {Wen}, \citenamefont {Huang}, \citenamefont {Xiang},
  \citenamefont {Li},\ and\ \citenamefont {Guo}}]{Zhao2023Optica}%
  \BibitemOpen
  \bibfield  {author} {\bibinfo {author} {\bibfnamefont {Y.-Y.}\ \bibnamefont
  {Zhao}}, \bibinfo {author} {\bibfnamefont {C.}~\bibnamefont {Zhang}},
  \bibinfo {author} {\bibfnamefont {S.}~\bibnamefont {Cheng}}, \bibinfo
  {author} {\bibfnamefont {X.}~\bibnamefont {Li}}, \bibinfo {author}
  {\bibfnamefont {Y.}~\bibnamefont {Guo}}, \bibinfo {author} {\bibfnamefont
  {B.-H.}\ \bibnamefont {Liu}}, \bibinfo {author} {\bibfnamefont {H.-Y.}\
  \bibnamefont {Ku}}, \bibinfo {author} {\bibfnamefont {S.-L.}\ \bibnamefont
  {Chen}}, \bibinfo {author} {\bibfnamefont {Q.}~\bibnamefont {Wen}}, \bibinfo
  {author} {\bibfnamefont {Y.-F.}\ \bibnamefont {Huang}}, \bibinfo {author}
  {\bibfnamefont {G.-Y.}\ \bibnamefont {Xiang}}, \bibinfo {author}
  {\bibfnamefont {C.-F.}\ \bibnamefont {Li}},\ and\ \bibinfo {author}
  {\bibfnamefont {G.-C.}\ \bibnamefont {Guo}},\ }\bibfield  {title} {\bibinfo
  {title} {Device-independent verification of {E}instein--{P}odolsky--{R}osen
  steering},\ }\href {https://doi.org/10.1364/OPTICA.456382} {\bibfield
  {journal} {\bibinfo  {journal} {Optica}\ }\textbf {\bibinfo {volume} {10}},\
  \bibinfo {pages} {66} (\bibinfo {year} {2023})}\BibitemShut {NoStop}%
\bibitem [{\citenamefont {Budroni}\ \emph {et~al.}(2022)\citenamefont
  {Budroni}, \citenamefont {Cabello}, \citenamefont {G\"uhne}, \citenamefont
  {Kleinmann},\ and\ \citenamefont {Larsson}}]{BudroniRMP2022}%
  \BibitemOpen
  \bibfield  {author} {\bibinfo {author} {\bibfnamefont {C.}~\bibnamefont
  {Budroni}}, \bibinfo {author} {\bibfnamefont {A.}~\bibnamefont {Cabello}},
  \bibinfo {author} {\bibfnamefont {O.}~\bibnamefont {G\"uhne}}, \bibinfo
  {author} {\bibfnamefont {M.}~\bibnamefont {Kleinmann}},\ and\ \bibinfo
  {author} {\bibfnamefont {J.-A.}\ \bibnamefont {Larsson}},\ }\bibfield
  {title} {\bibinfo {title} {Kochen-specker contextuality},\ }\href
  {https://doi.org/10.1103/RevModPhys.94.045007} {\bibfield  {journal}
  {\bibinfo  {journal} {Rev. Mod. Phys.}\ }\textbf {\bibinfo {volume} {94}},\
  \bibinfo {pages} {045007} (\bibinfo {year} {2022})}\BibitemShut {NoStop}%
\bibitem [{\citenamefont {Uola}\ \emph {et~al.}(2015)\citenamefont {Uola},
  \citenamefont {Budroni}, \citenamefont {G\"uhne},\ and\ \citenamefont
  {Pellonp\"a\"a}}]{Uola2015PRL}%
  \BibitemOpen
  \bibfield  {author} {\bibinfo {author} {\bibfnamefont {R.}~\bibnamefont
  {Uola}}, \bibinfo {author} {\bibfnamefont {C.}~\bibnamefont {Budroni}},
  \bibinfo {author} {\bibfnamefont {O.}~\bibnamefont {G\"uhne}},\ and\ \bibinfo
  {author} {\bibfnamefont {J.-P.}\ \bibnamefont {Pellonp\"a\"a}},\ }\bibfield
  {title} {\bibinfo {title} {One-to-one mapping between steering and joint
  measurability problems},\ }\href
  {https://doi.org/10.1103/PhysRevLett.115.230402} {\bibfield  {journal}
  {\bibinfo  {journal} {Phys. Rev. Lett.}\ }\textbf {\bibinfo {volume} {115}},\
  \bibinfo {pages} {230402} (\bibinfo {year} {2015})}\BibitemShut {NoStop}%
\bibitem [{\citenamefont {Kiukas}\ \emph {et~al.}(2017)\citenamefont {Kiukas},
  \citenamefont {Budroni}, \citenamefont {Uola},\ and\ \citenamefont
  {Pellonp\"a\"a}}]{Kiukas2017PRA}%
  \BibitemOpen
  \bibfield  {author} {\bibinfo {author} {\bibfnamefont {J.}~\bibnamefont
  {Kiukas}}, \bibinfo {author} {\bibfnamefont {C.}~\bibnamefont {Budroni}},
  \bibinfo {author} {\bibfnamefont {R.}~\bibnamefont {Uola}},\ and\ \bibinfo
  {author} {\bibfnamefont {J.-P.}\ \bibnamefont {Pellonp\"a\"a}},\ }\bibfield
  {title} {\bibinfo {title} {Continuous-variable steering and incompatibility
  via state-channel duality},\ }\href
  {https://doi.org/10.1103/PhysRevA.96.042331} {\bibfield  {journal} {\bibinfo
  {journal} {Phys. Rev. A}\ }\textbf {\bibinfo {volume} {96}},\ \bibinfo
  {pages} {042331} (\bibinfo {year} {2017})}\BibitemShut {NoStop}%
\bibitem [{\citenamefont {Datta}(2009)}]{Datta2009}%
  \BibitemOpen
  \bibfield  {author} {\bibinfo {author} {\bibfnamefont {N.}~\bibnamefont
  {Datta}},\ }\bibfield  {title} {\bibinfo {title} {Min- and max-relative
  entropies and a new entanglement monotone},\ }\href
  {https://doi.org/10.1109/TIT.2009.2018325} {\bibfield  {journal} {\bibinfo
  {journal} {IEEE Trans. Inf. Theory}\ }\textbf {\bibinfo {volume} {55}},\
  \bibinfo {pages} {2816} (\bibinfo {year} {2009})}\BibitemShut {NoStop}%
\bibitem [{\citenamefont {Gallego}\ and\ \citenamefont
  {Aolita}(2015)}]{Gallego2015PRX}%
  \BibitemOpen
  \bibfield  {author} {\bibinfo {author} {\bibfnamefont {R.}~\bibnamefont
  {Gallego}}\ and\ \bibinfo {author} {\bibfnamefont {L.}~\bibnamefont
  {Aolita}},\ }\bibfield  {title} {\bibinfo {title} {Resource theory of
  steering},\ }\href {https://doi.org/10.1103/PhysRevX.5.041008} {\bibfield
  {journal} {\bibinfo  {journal} {Phys. Rev. X}\ }\textbf {\bibinfo {volume}
  {5}},\ \bibinfo {pages} {041008} (\bibinfo {year} {2015})}\BibitemShut
  {NoStop}%
\bibitem [{\citenamefont {Skrzypczyk}\ \emph {et~al.}(2019)\citenamefont
  {Skrzypczyk}, \citenamefont {\ifmmode \check{S}\else
  \v{S}\fi{}upi\ifmmode~\acute{c}\else \'{c}\fi{}},\ and\ \citenamefont
  {Cavalcanti}}]{Paul2019PRL}%
  \BibitemOpen
  \bibfield  {author} {\bibinfo {author} {\bibfnamefont {P.}~\bibnamefont
  {Skrzypczyk}}, \bibinfo {author} {\bibfnamefont {I.}~\bibnamefont {\ifmmode
  \check{S}\else \v{S}\fi{}upi\ifmmode~\acute{c}\else \'{c}\fi{}}},\ and\
  \bibinfo {author} {\bibfnamefont {D.}~\bibnamefont {Cavalcanti}},\ }\bibfield
   {title} {\bibinfo {title} {All sets of incompatible measurements give an
  advantage in quantum state discrimination},\ }\href
  {https://doi.org/10.1103/PhysRevLett.122.130403} {\bibfield  {journal}
  {\bibinfo  {journal} {Phys. Rev. Lett.}\ }\textbf {\bibinfo {volume} {122}},\
  \bibinfo {pages} {130403} (\bibinfo {year} {2019})}\BibitemShut {NoStop}%
\bibitem [{\citenamefont {Yadin}\ \emph {et~al.}(2021)\citenamefont {Yadin},
  \citenamefont {Fadel},\ and\ \citenamefont {Gessner}}]{Yadin2021NC}%
  \BibitemOpen
  \bibfield  {author} {\bibinfo {author} {\bibfnamefont {B.}~\bibnamefont
  {Yadin}}, \bibinfo {author} {\bibfnamefont {M.}~\bibnamefont {Fadel}},\ and\
  \bibinfo {author} {\bibfnamefont {M.}~\bibnamefont {Gessner}},\ }\bibfield
  {title} {\bibinfo {title} {Metrological complementarity reveals the
  {E}instein-{P}odolsky-{R}osen paradox},\ }\href
  {https://doi.org/10.1038/s41467-021-22353-3} {\bibfield  {journal} {\bibinfo
  {journal} {Nat. Commun.}\ }\textbf {\bibinfo {volume} {12}},\ \bibinfo
  {pages} {2410} (\bibinfo {year} {2021})}\BibitemShut {NoStop}%
\bibitem [{\citenamefont {Designolle}\ \emph {et~al.}(2019)\citenamefont
  {Designolle}, \citenamefont {Farkas},\ and\ \citenamefont
  {Kaniewski}}]{Designolle2019}%
  \BibitemOpen
  \bibfield  {author} {\bibinfo {author} {\bibfnamefont {S.}~\bibnamefont
  {Designolle}}, \bibinfo {author} {\bibfnamefont {M.}~\bibnamefont {Farkas}},\
  and\ \bibinfo {author} {\bibfnamefont {J.}~\bibnamefont {Kaniewski}},\
  }\bibfield  {title} {\bibinfo {title} {Incompatibility robustness of quantum
  measurements: a unified framework},\ }\href
  {https://doi.org/10.1088/1367-2630/ab5020} {\bibfield  {journal} {\bibinfo
  {journal} {New J. Phys.}\ }\textbf {\bibinfo {volume} {21}},\ \bibinfo
  {pages} {113053} (\bibinfo {year} {2019})}\BibitemShut {NoStop}%
\bibitem [{\citenamefont {Takagi}\ and\ \citenamefont
  {Regula}(2019)}]{Takagi2019}%
  \BibitemOpen
  \bibfield  {author} {\bibinfo {author} {\bibfnamefont {R.}~\bibnamefont
  {Takagi}}\ and\ \bibinfo {author} {\bibfnamefont {B.}~\bibnamefont
  {Regula}},\ }\bibfield  {title} {\bibinfo {title} {General resource theories
  in quantum mechanics and beyond: Operational characterization via
  discrimination tasks},\ }\href {https://doi.org/10.1103/PhysRevX.9.031053}
  {\bibfield  {journal} {\bibinfo  {journal} {Phys. Rev. X}\ }\textbf {\bibinfo
  {volume} {9}},\ \bibinfo {pages} {031053} (\bibinfo {year}
  {2019})}\BibitemShut {NoStop}%
\bibitem [{\citenamefont {Hsieh}\ \emph
  {et~al.}(2016{\natexlab{a}})\citenamefont {Hsieh}, \citenamefont {Liang},\
  and\ \citenamefont {Lee}}]{Hsieh2016PRA}%
  \BibitemOpen
  \bibfield  {author} {\bibinfo {author} {\bibfnamefont {C.-Y.}\ \bibnamefont
  {Hsieh}}, \bibinfo {author} {\bibfnamefont {Y.-C.}\ \bibnamefont {Liang}},\
  and\ \bibinfo {author} {\bibfnamefont {R.-K.}\ \bibnamefont {Lee}},\
  }\bibfield  {title} {\bibinfo {title} {Quantum steerability:
  Characterization, quantification, superactivation, and unbounded
  amplification},\ }\href {https://doi.org/10.1103/PhysRevA.94.062120}
  {\bibfield  {journal} {\bibinfo  {journal} {Phys. Rev. A}\ }\textbf {\bibinfo
  {volume} {94}},\ \bibinfo {pages} {062120} (\bibinfo {year}
  {2016}{\natexlab{a}})}\BibitemShut {NoStop}%
\bibitem [{\citenamefont {Ku}\ \emph {et~al.}(2018)\citenamefont {Ku},
  \citenamefont {Chen}, \citenamefont {Budroni}, \citenamefont {Miranowicz},
  \citenamefont {Chen},\ and\ \citenamefont {Nori}}]{Ku2018PRA}%
  \BibitemOpen
  \bibfield  {author} {\bibinfo {author} {\bibfnamefont {H.-Y.}\ \bibnamefont
  {Ku}}, \bibinfo {author} {\bibfnamefont {S.-L.}\ \bibnamefont {Chen}},
  \bibinfo {author} {\bibfnamefont {C.}~\bibnamefont {Budroni}}, \bibinfo
  {author} {\bibfnamefont {A.}~\bibnamefont {Miranowicz}}, \bibinfo {author}
  {\bibfnamefont {Y.-N.}\ \bibnamefont {Chen}},\ and\ \bibinfo {author}
  {\bibfnamefont {F.}~\bibnamefont {Nori}},\ }\bibfield  {title} {\bibinfo
  {title} {Einstein-podolsky-rosen steering: Its geometric quantification and
  witness},\ }\href {https://doi.org/10.1103/PhysRevA.97.022338} {\bibfield
  {journal} {\bibinfo  {journal} {Phys. Rev. A}\ }\textbf {\bibinfo {volume}
  {97}},\ \bibinfo {pages} {022338} (\bibinfo {year} {2018})}\BibitemShut
  {NoStop}%
\bibitem [{\citenamefont {Tendick}\ \emph {et~al.}(2023)\citenamefont
  {Tendick}, \citenamefont {Kliesch}, \citenamefont {Kampermann},\ and\
  \citenamefont {Bru{\ss{}}}}]{Tendick2023Quantum}%
  \BibitemOpen
  \bibfield  {author} {\bibinfo {author} {\bibfnamefont {L.}~\bibnamefont
  {Tendick}}, \bibinfo {author} {\bibfnamefont {M.}~\bibnamefont {Kliesch}},
  \bibinfo {author} {\bibfnamefont {H.}~\bibnamefont {Kampermann}},\ and\
  \bibinfo {author} {\bibfnamefont {D.}~\bibnamefont {Bru{\ss{}}}},\ }\bibfield
   {title} {\bibinfo {title} {Distance-based resource quantification for sets
  of quantum measurements},\ }\href
  {https://doi.org/10.22331/q-2023-05-15-1003} {\bibfield  {journal} {\bibinfo
  {journal} {{Quantum}}\ }\textbf {\bibinfo {volume} {7}},\ \bibinfo {pages}
  {1003} (\bibinfo {year} {2023})}\BibitemShut {NoStop}%
\bibitem [{\citenamefont {Haapasalo}(2015)}]{Haapasalo2015}%
  \BibitemOpen
  \bibfield  {author} {\bibinfo {author} {\bibfnamefont {E.}~\bibnamefont
  {Haapasalo}},\ }\bibfield  {title} {\bibinfo {title} {Robustness of
  incompatibility for quantum devices},\ }\href
  {https://doi.org/10.1088/1751-8113/48/25/255303} {\bibfield  {journal}
  {\bibinfo  {journal} {J. Phys. A: Math. Theor.}\ }\textbf {\bibinfo {volume}
  {48}},\ \bibinfo {pages} {255303} (\bibinfo {year} {2015})}\BibitemShut
  {NoStop}%
\bibitem [{\citenamefont {Buscemi}\ \emph {et~al.}(2020)\citenamefont
  {Buscemi}, \citenamefont {Chitambar},\ and\ \citenamefont
  {Zhou}}]{Buscemi2020PRL}%
  \BibitemOpen
  \bibfield  {author} {\bibinfo {author} {\bibfnamefont {F.}~\bibnamefont
  {Buscemi}}, \bibinfo {author} {\bibfnamefont {E.}~\bibnamefont {Chitambar}},\
  and\ \bibinfo {author} {\bibfnamefont {W.}~\bibnamefont {Zhou}},\ }\bibfield
  {title} {\bibinfo {title} {Complete resource theory of quantum
  incompatibility as quantum programmability},\ }\href
  {https://doi.org/10.1103/PhysRevLett.124.120401} {\bibfield  {journal}
  {\bibinfo  {journal} {Phys. Rev. Lett.}\ }\textbf {\bibinfo {volume} {124}},\
  \bibinfo {pages} {120401} (\bibinfo {year} {2020})}\BibitemShut {NoStop}%
\bibitem [{\citenamefont {Cavalcanti}\ and\ \citenamefont
  {Skrzypczyk}(2016{\natexlab{b}})}]{Cavalcanti2016PRA}%
  \BibitemOpen
  \bibfield  {author} {\bibinfo {author} {\bibfnamefont {D.}~\bibnamefont
  {Cavalcanti}}\ and\ \bibinfo {author} {\bibfnamefont {P.}~\bibnamefont
  {Skrzypczyk}},\ }\bibfield  {title} {\bibinfo {title} {Quantitative relations
  between measurement incompatibility, quantum steering, and nonlocality},\
  }\href {https://doi.org/10.1103/PhysRevA.93.052112} {\bibfield  {journal}
  {\bibinfo  {journal} {Phys. Rev. A}\ }\textbf {\bibinfo {volume} {93}},\
  \bibinfo {pages} {052112} (\bibinfo {year} {2016}{\natexlab{b}})}\BibitemShut
  {NoStop}%
\bibitem [{\citenamefont {Chen}\ \emph {et~al.}(2018)\citenamefont {Chen},
  \citenamefont {Budroni}, \citenamefont {Liang},\ and\ \citenamefont
  {Chen}}]{Chen2018PRA}%
  \BibitemOpen
  \bibfield  {author} {\bibinfo {author} {\bibfnamefont {S.-L.}\ \bibnamefont
  {Chen}}, \bibinfo {author} {\bibfnamefont {C.}~\bibnamefont {Budroni}},
  \bibinfo {author} {\bibfnamefont {Y.-C.}\ \bibnamefont {Liang}},\ and\
  \bibinfo {author} {\bibfnamefont {Y.-N.}\ \bibnamefont {Chen}},\ }\bibfield
  {title} {\bibinfo {title} {Exploring the framework of assemblage moment
  matrices and its applications in device-independent characterizations},\
  }\href {https://doi.org/10.1103/PhysRevA.98.042127} {\bibfield  {journal}
  {\bibinfo  {journal} {Phys. Rev. A}\ }\textbf {\bibinfo {volume} {98}},\
  \bibinfo {pages} {042127} (\bibinfo {year} {2018})}\BibitemShut {NoStop}%
\bibitem [{\citenamefont {Hsieh}\ \emph {et~al.}(2022)\citenamefont {Hsieh},
  \citenamefont {Lostaglio},\ and\ \citenamefont {Ac\'{\i}n}}]{Hsieh2022PRR}%
  \BibitemOpen
  \bibfield  {author} {\bibinfo {author} {\bibfnamefont {C.-Y.}\ \bibnamefont
  {Hsieh}}, \bibinfo {author} {\bibfnamefont {M.}~\bibnamefont {Lostaglio}},\
  and\ \bibinfo {author} {\bibfnamefont {A.}~\bibnamefont {Ac\'{\i}n}},\
  }\bibfield  {title} {\bibinfo {title} {Quantum channel marginal problem},\
  }\href {https://doi.org/10.1103/PhysRevResearch.4.013249} {\bibfield
  {journal} {\bibinfo  {journal} {Phys. Rev. Res.}\ }\textbf {\bibinfo {volume}
  {4}},\ \bibinfo {pages} {013249} (\bibinfo {year} {2022})}\BibitemShut
  {NoStop}%
\bibitem [{\citenamefont {Hsieh}\ \emph {et~al.}()\citenamefont {Hsieh},
  \citenamefont {Tabia}, \citenamefont {Yin},\ and\ \citenamefont
  {Liang}}]{Hsieh2023-2}%
  \BibitemOpen
  \bibfield  {author} {\bibinfo {author} {\bibfnamefont {C.-Y.}\ \bibnamefont
  {Hsieh}}, \bibinfo {author} {\bibfnamefont {G.~N.~M.}\ \bibnamefont {Tabia}},
  \bibinfo {author} {\bibfnamefont {Y.-C.}\ \bibnamefont {Yin}},\ and\ \bibinfo
  {author} {\bibfnamefont {Y.-C.}\ \bibnamefont {Liang}},\ }\href@noop {}
  {\bibinfo {title} {Resoruce marginal problems}},\ \Eprint
  {https://arxiv.org/abs/2202.03523} {arXiv:2202.03523} \BibitemShut {NoStop}%
\bibitem [{\citenamefont {Skrzypczyk}\ and\ \citenamefont
  {Cavalcanti}(2023)}]{SDP-textbook}%
  \BibitemOpen
  \bibfield  {author} {\bibinfo {author} {\bibfnamefont {P.}~\bibnamefont
  {Skrzypczyk}}\ and\ \bibinfo {author} {\bibfnamefont {D.}~\bibnamefont
  {Cavalcanti}},\ }\href {https://doi.org/10.1088/978-0-7503-3343-6} {\emph
  {\bibinfo {title} {Semidefinite Programming in Quantum Information
  Science}}},\ 2053-2563\ (\bibinfo  {publisher} {IOP Publishing},\ \bibinfo
  {year} {2023})\BibitemShut {NoStop}%
\bibitem [{\citenamefont {Beyer}\ \emph {et~al.}(2019)\citenamefont {Beyer},
  \citenamefont {Luoma},\ and\ \citenamefont {Strunz}}]{BeyerPRL2019}%
  \BibitemOpen
  \bibfield  {author} {\bibinfo {author} {\bibfnamefont {K.}~\bibnamefont
  {Beyer}}, \bibinfo {author} {\bibfnamefont {K.}~\bibnamefont {Luoma}},\ and\
  \bibinfo {author} {\bibfnamefont {W.~T.}\ \bibnamefont {Strunz}},\ }\bibfield
   {title} {\bibinfo {title} {Steering heat engines: A truly quantum maxwell
  demon},\ }\href {https://doi.org/10.1103/PhysRevLett.123.250606} {\bibfield
  {journal} {\bibinfo  {journal} {Phys. Rev. Lett.}\ }\textbf {\bibinfo
  {volume} {123}},\ \bibinfo {pages} {250606} (\bibinfo {year}
  {2019})}\BibitemShut {NoStop}%
\bibitem [{\citenamefont {Hsieh}(2020)}]{Hsieh2020}%
  \BibitemOpen
  \bibfield  {author} {\bibinfo {author} {\bibfnamefont {C.-Y.}\ \bibnamefont
  {Hsieh}},\ }\bibfield  {title} {\bibinfo {title} {Resource preservability},\
  }\href {https://doi.org/10.22331/q-2020-03-19-244} {\bibfield  {journal}
  {\bibinfo  {journal} {{Quantum}}\ }\textbf {\bibinfo {volume} {4}},\ \bibinfo
  {pages} {244} (\bibinfo {year} {2020})}\BibitemShut {NoStop}%
\bibitem [{\citenamefont {Hsieh}(2021)}]{Hsieh2021PRXQ}%
  \BibitemOpen
  \bibfield  {author} {\bibinfo {author} {\bibfnamefont {C.-Y.}\ \bibnamefont
  {Hsieh}},\ }\bibfield  {title} {\bibinfo {title} {Communication, dynamical
  resource theory, and thermodynamics},\ }\href
  {https://doi.org/10.1103/PRXQuantum.2.020318} {\bibfield  {journal} {\bibinfo
   {journal} {PRX Quantum}\ }\textbf {\bibinfo {volume} {2}},\ \bibinfo {pages}
  {020318} (\bibinfo {year} {2021})}\BibitemShut {NoStop}%
\bibitem [{\citenamefont {Hsieh}()}]{Hsieh2022}%
  \BibitemOpen
  \bibfield  {author} {\bibinfo {author} {\bibfnamefont {C.-Y.}\ \bibnamefont
  {Hsieh}},\ }\href@noop {} {\bibinfo {title} {Thermodynamic criterion of
  transmitting classical information}},\ \Eprint
  {https://arxiv.org/abs/2201.12110} {arXiv:2201.12110} \BibitemShut {NoStop}%
\bibitem [{\citenamefont {Stratton}\ \emph {et~al.}()\citenamefont {Stratton},
  \citenamefont {Hsieh},\ and\ \citenamefont {Skrzypczyk}}]{Stratton2023}%
  \BibitemOpen
  \bibfield  {author} {\bibinfo {author} {\bibfnamefont {B.}~\bibnamefont
  {Stratton}}, \bibinfo {author} {\bibfnamefont {C.-Y.}\ \bibnamefont
  {Hsieh}},\ and\ \bibinfo {author} {\bibfnamefont {P.}~\bibnamefont
  {Skrzypczyk}},\ }\href@noop {} {\bibinfo {title} {The dynamical resource
  theory of informational non-equilibrium}},\ \Eprint
  {https://arxiv.org/abs/2306.16848} {arXiv:2306.16848} \BibitemShut {NoStop}%
\bibitem [{\citenamefont {Costa}\ and\ \citenamefont
  {Angelo}(2016)}]{CostaPRA2016}%
  \BibitemOpen
  \bibfield  {author} {\bibinfo {author} {\bibfnamefont {A.~C.~S.}\
  \bibnamefont {Costa}}\ and\ \bibinfo {author} {\bibfnamefont {R.~M.}\
  \bibnamefont {Angelo}},\ }\bibfield  {title} {\bibinfo {title}
  {Quantification of einstein-podolsky-rosen steering for two-qubit states},\
  }\href {https://doi.org/10.1103/PhysRevA.93.020103} {\bibfield  {journal}
  {\bibinfo  {journal} {Phys. Rev. A}\ }\textbf {\bibinfo {volume} {93}},\
  \bibinfo {pages} {020103} (\bibinfo {year} {2016})}\BibitemShut {NoStop}%
\bibitem [{\citenamefont {Skrzypczyk}\ and\ \citenamefont
  {Cavalcanti}(2018{\natexlab{b}})}]{SkrzypczykPRL2018}%
  \BibitemOpen
  \bibfield  {author} {\bibinfo {author} {\bibfnamefont {P.}~\bibnamefont
  {Skrzypczyk}}\ and\ \bibinfo {author} {\bibfnamefont {D.}~\bibnamefont
  {Cavalcanti}},\ }\bibfield  {title} {\bibinfo {title} {Maximal randomness
  generation from steering inequality violations using qudits},\ }\href
  {https://doi.org/10.1103/PhysRevLett.120.260401} {\bibfield  {journal}
  {\bibinfo  {journal} {Phys. Rev. Lett.}\ }\textbf {\bibinfo {volume} {120}},\
  \bibinfo {pages} {260401} (\bibinfo {year} {2018}{\natexlab{b}})}\BibitemShut
  {NoStop}%
\bibitem [{\citenamefont {Ji}\ \emph {et~al.}(2022)\citenamefont {Ji},
  \citenamefont {Chai}, \citenamefont {Wang}, \citenamefont {Guo},
  \citenamefont {Rong}, \citenamefont {Shi}, \citenamefont {Ren}, \citenamefont
  {Wang},\ and\ \citenamefont {Du}}]{JiPRL2022}%
  \BibitemOpen
  \bibfield  {author} {\bibinfo {author} {\bibfnamefont {W.}~\bibnamefont
  {Ji}}, \bibinfo {author} {\bibfnamefont {Z.}~\bibnamefont {Chai}}, \bibinfo
  {author} {\bibfnamefont {M.}~\bibnamefont {Wang}}, \bibinfo {author}
  {\bibfnamefont {Y.}~\bibnamefont {Guo}}, \bibinfo {author} {\bibfnamefont
  {X.}~\bibnamefont {Rong}}, \bibinfo {author} {\bibfnamefont {F.}~\bibnamefont
  {Shi}}, \bibinfo {author} {\bibfnamefont {C.}~\bibnamefont {Ren}}, \bibinfo
  {author} {\bibfnamefont {Y.}~\bibnamefont {Wang}},\ and\ \bibinfo {author}
  {\bibfnamefont {J.}~\bibnamefont {Du}},\ }\bibfield  {title} {\bibinfo
  {title} {Spin quantum heat engine quantified by quantum steering},\ }\href
  {https://doi.org/10.1103/PhysRevLett.128.090602} {\bibfield  {journal}
  {\bibinfo  {journal} {Phys. Rev. Lett.}\ }\textbf {\bibinfo {volume} {128}},\
  \bibinfo {pages} {090602} (\bibinfo {year} {2022})}\BibitemShut {NoStop}%
\bibitem [{\citenamefont {Lopetegui}\ \emph {et~al.}(2022)\citenamefont
  {Lopetegui}, \citenamefont {Gessner}, \citenamefont {Fadel}, \citenamefont
  {Treps},\ and\ \citenamefont {Walschaers}}]{LopeteguiPRXQ2022}%
  \BibitemOpen
  \bibfield  {author} {\bibinfo {author} {\bibfnamefont {C.~E.}\ \bibnamefont
  {Lopetegui}}, \bibinfo {author} {\bibfnamefont {M.}~\bibnamefont {Gessner}},
  \bibinfo {author} {\bibfnamefont {M.}~\bibnamefont {Fadel}}, \bibinfo
  {author} {\bibfnamefont {N.}~\bibnamefont {Treps}},\ and\ \bibinfo {author}
  {\bibfnamefont {M.}~\bibnamefont {Walschaers}},\ }\bibfield  {title}
  {\bibinfo {title} {Homodyne detection of non-gaussian quantum steering},\
  }\href {https://doi.org/10.1103/PRXQuantum.3.030347} {\bibfield  {journal}
  {\bibinfo  {journal} {PRX Quantum}\ }\textbf {\bibinfo {volume} {3}},\
  \bibinfo {pages} {030347} (\bibinfo {year} {2022})}\BibitemShut {NoStop}%
\bibitem [{\citenamefont {Hsieh}\ \emph
  {et~al.}(2016{\natexlab{b}})\citenamefont {Hsieh}, \citenamefont {Liang},\
  and\ \citenamefont {Lee}}]{Hsieh2016}%
  \BibitemOpen
  \bibfield  {author} {\bibinfo {author} {\bibfnamefont {C.-Y.}\ \bibnamefont
  {Hsieh}}, \bibinfo {author} {\bibfnamefont {Y.-C.}\ \bibnamefont {Liang}},\
  and\ \bibinfo {author} {\bibfnamefont {R.-K.}\ \bibnamefont {Lee}},\
  }\bibfield  {title} {\bibinfo {title} {Quantum steerability:
  Characterization, quantification, superactivation, and unbounded
  amplification},\ }\href {https://doi.org/10.1103/PhysRevA.94.062120}
  {\bibfield  {journal} {\bibinfo  {journal} {Phys. Rev. A}\ }\textbf {\bibinfo
  {volume} {94}},\ \bibinfo {pages} {062120} (\bibinfo {year}
  {2016}{\natexlab{b}})}\BibitemShut {NoStop}%
\bibitem [{\citenamefont {Quintino}\ \emph {et~al.}(2016)\citenamefont
  {Quintino}, \citenamefont {Brunner},\ and\ \citenamefont
  {Huber}}]{Quintino2016}%
  \BibitemOpen
  \bibfield  {author} {\bibinfo {author} {\bibfnamefont {M.~T.}\ \bibnamefont
  {Quintino}}, \bibinfo {author} {\bibfnamefont {N.}~\bibnamefont {Brunner}},\
  and\ \bibinfo {author} {\bibfnamefont {M.}~\bibnamefont {Huber}},\ }\bibfield
   {title} {\bibinfo {title} {Superactivation of quantum steering},\ }\href
  {https://doi.org/10.1103/PhysRevA.94.062123} {\bibfield  {journal} {\bibinfo
  {journal} {Phys. Rev. A}\ }\textbf {\bibinfo {volume} {94}},\ \bibinfo
  {pages} {062123} (\bibinfo {year} {2016})}\BibitemShut {NoStop}%
\bibitem [{\citenamefont {Ku}\ \emph {et~al.}(2023)\citenamefont {Ku},
  \citenamefont {Lee}, \citenamefont {Lai}, \citenamefont {Lin},\ and\
  \citenamefont {Chen}}]{KuPRA2023}%
  \BibitemOpen
  \bibfield  {author} {\bibinfo {author} {\bibfnamefont {H.-Y.}\ \bibnamefont
  {Ku}}, \bibinfo {author} {\bibfnamefont {K.-Y.}\ \bibnamefont {Lee}},
  \bibinfo {author} {\bibfnamefont {P.-R.}\ \bibnamefont {Lai}}, \bibinfo
  {author} {\bibfnamefont {J.-D.}\ \bibnamefont {Lin}},\ and\ \bibinfo {author}
  {\bibfnamefont {Y.-N.}\ \bibnamefont {Chen}},\ }\bibfield  {title} {\bibinfo
  {title} {Coherent activation of a steerability-breaking channel},\ }\href
  {https://doi.org/10.1103/PhysRevA.107.042415} {\bibfield  {journal} {\bibinfo
   {journal} {Phys. Rev. A}\ }\textbf {\bibinfo {volume} {107}},\ \bibinfo
  {pages} {042415} (\bibinfo {year} {2023})}\BibitemShut {NoStop}%
\end{thebibliography}%

\end{document}